\documentclass[12pt]{article}
\pdfoutput=1
\usepackage[a4paper,left=2cm,right=1.5cm,top=2cm,bottom=2cm]{geometry}
\usepackage{amsmath,amssymb,bm}
\usepackage{graphicx}
\usepackage{natbib}
\usepackage[font=small,labelfont=bf]{caption}
\usepackage[hyphens]{url}
\usepackage{booktabs}

\usepackage{setspace}
\usepackage{amsthm}
\newtheorem{theorem}{Theorem}
\newtheorem{prop}{Proposition}

\newtheorem{lemma}{Lemma}

\usepackage{algorithm}
\usepackage{algorithmic}

\DeclareMathOperator*{\argmin}{argmin}

\title{Epidemic changepoint detection \\ in the presence of nuisance changes}

\author{Julius Juodakis, Stephen Marsland \\
        \normalsize{School of Mathematics and Statistics, Victoria University of Wellington,
        New Zealand} \\
        \normalsize{Corresponding author: \texttt{julius.juodakis@sms.vuw.ac.nz}}
}

\date{}

\begin{document}
\bibliographystyle{rss}

\maketitle
\abstract{
    Many time series problems feature epidemic changes -- segments where a parameter deviates from a background baseline. The number and location of such changes can be estimated in a principled way by existing detection methods, providing that the background level is stable and known. However, practical data often contains nuisance changes in background level, which interfere with standard estimation techniques. Furthermore, such changes often differ from the target segments only in duration, and appear as false alarms in the detection results.

    To solve these issues, we propose a two-level detector that models and separates nuisance and signal changes. As part of this method, we developed a new, efficient approach to simultaneously estimate unknown, but fixed, background level and detect epidemic changes. The analytic and computational properties of the proposed methods are established, including consistency and convergence.
    We demonstrate via simulations that our two-level detector provides accurate estimation of changepoints under a nuisance process, while other state-of-the-art detectors fail. Using real-world genomic and demographic datasets, we demonstrate that our method can identify and localise target events while separating out seasonal variations and experimental artefacts.
}


\section{Introduction}
The problem of identifying when the probability distribution of a time series changes -- changepoint detection -- has been studied since the middle of the 20th century. Early developments, such as Shewhart's control charts \citep{Shewhart1930} and Page's CUSUM test \citep{Page1954}, stemmed from operations research. However, as automatic sensing systems and continuous data collection became more common, many new use cases for changepoint detection have arisen, such as seismic events \citep{Li2016}, epidemic outbreaks \citep{Texier2016}, physiological signals \citep{Vaisman2010}, market movements \citep{Bracker1999}, and network traffic spikes \citep{Hochenbaum2017}. Stimulated by such practical interest, the growth of corresponding statistical theory has been rapid, as reviewed in e.g. \citet{Aminikhanghahi2016} and \citet{Truong2018}.

Different applications pose different statistical challenges. If a single drastic change may be expected, such as when detecting machine failure, the goal is to find a method with minimal time to detection and a controlled false alarm rate \citep{Lau2019}. In many other applications the number of changepoints is unknown a priori; the challenge is then to identify the change locations regardless, preferably in a computationally efficient way. Some problems, such as peak detection in sound \citep{Aminikhanghahi2016} or genetic data \citep{HockingBio2016}, feature epidemic segments -- changepoints followed by a return to the background level -- and incorporating this constraint can improve detection or simplify post-processing of ouputs.

However, current detection methods that do incorporate a background level assume it to be stable throughout the dataset \citep[e.g.][]{CAPA2018, aPELT2019}. This is not realistic in several common applications. In genetic analyses such as measurements of protein binding along DNA there may be large regions where the observed background level is shifted due to structural variation in the genome or technical artefacts \citep{MACS2008}. Similarly, a standard task in sound processing is to detect speech in the presence of dynamic background chatter, so-called babble noise \citep{Aurora2000}. In various datasets from epidemiology or climatology, such as wave height measurements \citep{KillickPELT}, seasonal effects are observed as recurring background changes and will interfere with detection of shorter events. Methods that assume a constant background will be inaccurate in these cases, while ignoring the epidemic structure entirely would cost detection power and complicate the interpretation of outputs.

Our goal is to develop a general method for detecting epidemic changepoints in the presence of nuisance changes in the background. Furthermore, we assume that the only available feature for distinguishing the two types of segments is their duration: this would allow analysing the examples above, which share the property that the nuisance process is slower. The most similar research to ours is that of \citet{Lau2019} for detecting failure of a machine that can be switched on, and thus undergo an irrelevant change in the background level. However, the setting there concerned a single change with specified background and nuisance distributions; in contrast, we are motivated by the case when multiple changes may be present, with only duration distinguishing their types. 
Detection then requires two novel developments: 1) rapid estimation of local background level, 2) modelling and distinguishing the two types of potentially overlapping segments.

These developments are presented in this paper as follows: after a background section we provide a new algorithm that simultaneously detects epidemic changepoints and estimates the unknown background level (Section \ref{sectshort}). The convergence and consistency of this algorithm are proved in the Supplementary Material. While this algorithm is of its own interest, we use it to build a detector that allows local variations in the background, i.e., nuisance changes, in Section \ref{sectmethod}, again with proven properties. In Section \ref{sectexperiments} we test various properties of the algorithms with simulations before showing how the proposed nuisance-robust detector can be applied in two problems: detecting histone modifications in human genome, while ignoring structural variations, and detecting the effects of the COVID-19 pandemic in Spanish mortality data, robustly to seasonal effects. Compared to state-of-the-art methods, the proposed detector produced lower false-alarm rates (or more parsimonious models), while retaining accurate detection of true signal peaks.

\section{Background}

In the general changepoint detection setup, the data consists of a sequence of observations $x_{0:n} = \{x_0, x_1, \dots, x_n\}$, split by $k$ changepoints $0<\tau_1<\tau_2<\dots\tau_k<n$ into $k+1$ segments. Within each segment the observations are drawn from a distribution specific to that segment; these are often assumed to belong to some fixed family $f_{\theta}(x)$, with different values of parameter $\theta$ for each segment. The most common and historically earliest example is the change in mean of a Gaussian (see \citet{Truong2018}), i.e., for each $t \in [\tau_i, \tau_{i+1})$, $x_t \sim \mathcal{N}(\mu_i, \sigma^2)$, for known, fixed $\sigma^2$. (We assume $\theta \in \mathbb{R}^1$ to keep notation clearer, but multidimensional problems are also common, such as when both the mean and variance of a Gaussian are subject to change.)

While early research focused on optimal hypothesis testing of a single change against no change in a parameter, modern models aim to estimate the position of all changepoints $\{ \tau_i \}$ in a given time series, and optionally the vector of segment parameters $\bm{\theta}$. A common approach is to use a penalised cost: choose a segment cost function $C(x_{a:b}; \theta)$, usually based on negative log-likelihood of the observations $x_{a:b}$ given $\theta$, and a penalty $p(k)$ for the number of changepoints $k$. The full cost of $x_{\tau_0:\tau_k}$ (where $\tau_0 = 0$ and $\tau_k = n$) is then:
\begin{equation}
    F(n; \bm{\tau}, \bm{\theta}, k) = \sum_{i=0}^{k-1} C(x_{\tau_i+1:\tau_{i+1}}; \theta_i) + p(k).
\label{eq:F}
\end{equation}
Changepoint number and positions are estimated by finding $F(n) = \min F(n; \bm{\tau}, \bm{\theta}, k)$. Such estimation has been shown to be consistent for a range of different data generation models \citep{CAPA2018, Zheng2019}.

For a discrete problem such as this, computation of the true minimum is not trivial -- a na\"ive brute force approach would require $\mathcal{O}(2^n)$ tests. Approaches to reducing this cost fall in two broad classes: 1) simplifying the search algorithm by memoisation and pruning of paths \citep{JacksonOP, KillickPELT, Rigaill2010}; 2) using greedy methods to find approximate solutions faster \citep{Fryzlewicz2014, Baranowski2019}. In both classes, there are methods that can consistently estimate multiple changepoints in linear time under certain conditions.

The first category is of more relevance here. It is based on the Optimal Partitioning (OP) algorithm \citep{JacksonOP}. Let the data be partitioned into discrete blocks $B_i: \bigcup_i B_i = \{x\}_n$, so $B_i \cap B_j = \emptyset, \forall i \neq j$. A function $V$ that maps each set of blocks $P_j=\{B_i\}$ to a cost is block-additive if:
\begin{equation}
    \forall P_1, P_2, V(P_1 \cup P_2) = V(P_1) + V(P_2).
    \label{eq:blockadd}
\end{equation}
If each segment incurs a fixed penalty $\beta = p(k)/k$, then the function $F$ defined in \eqref{eq:F} is block-additive over segments, and can be defined recursively as:
\[ F(s) = \min_t \left( F(t) + C(x_{t+1:s}) + \beta \right). \]
In OP, this cost is calculated for each $s \le n$, and thus its minimisation requires $\mathcal{O}(n^2)$ evaluations of $C$.
Furthermore, when the cost function $C$ is such that for all $a \le b < c$:
\begin{equation}
    C(x_{a:c}) \ge C(x_{a:b}) + C(x_{b+1:c})
    \label{eq:pruning}
\end{equation}
then at each $s$, it can be quickly determined that some candidate segmentations cannot be ``rescued'' by further segments, and so they can be pruned from the search space. This approach reduces the complexity to $\mathcal{O}(n)$ in the best case. It gave rise to the family of algorithms called PELT, Pruned Exact Linear Time \citep{KillickPELT, aPELT2019}. Note that OP and PELT do not rely on any probabilistic assumptions and find the exact same minimum as the global exponential-complexity search. Since the introduction of PELT, a number of alternative pruning schemes have been developed \citep{Maidstone2014}.

A variation on the basic changepoint detection problem is the identification of epidemic changepoints, when a change in regime appears for a finite time and then the process returns to the background level.
The concept of a background distribution $f_B$ is introduced, and the segments are defined by pairs of changepoints $s_i, e_i$, outside of which the data is drawn from the background model. The data model then becomes:
\begin{equation}
    f(x_t) = \begin{cases} f_S(x_t; \theta_i) & \text{ if } \exists i: s_i \le t \le e_i \\
    f_B(x_t) & \text{ otherwise}. \end{cases}
    \label{eq:epidemmodel}
\end{equation}
An equivalent of the CUSUM test for the epidemic situation was first proposed by \citet{Levin1985}, and a number of methods for multiple changepoint detection in this setting have been proposed \citep{Olshen2004, CAPA2018, aPELT2019}. These use a cost function that includes a separate term $C^0$ for the background points:
\begin{equation}
    F(n; \{(s_i, e_i)\}_k, \bm{\theta}, k) = \sum_{i=1}^k C(x_{s_i:e_i}; \theta_i) + C^0(\{x_t : t \notin \bigcup_i[s_i;e_i]\}; \theta_0) + p(k)
    \label{fullcostepid}
\end{equation}

A common choice for the background distribution is some particular ``null'' case of the segment distribution family, so that $f_B(x) = f_S(x; \theta_0)$ and $C^0(\cdot) = C(\cdot; \theta_0)$. 
However, while the value of $\theta_0$ is known in some settings (such as when presented with two copies of DNA), in other cases it may need to be estimated. Since this parameter is shared across all background points, the cost function is no longer block-additive as in \eqref{eq:blockadd}, and algorithms such as OP and PELT cannot be directly applied. 

One solution is to substitute the unknown parameter with some robust estimate of it, based on the unsegmented series $x_{0:n}$. The success of the resulting changepoint estimation then relies on this estimate being sufficiently close to the true value, and so the non-background data fraction must be small \citep{CAPA2018}. This is unlikely to be true in our motivating case, when nuisance changes in the background level are possible.

Another option is to define
\[ F(n; \theta_0) = \min_{k,\{(s_i, e_i)\}_k, \{\theta_i\}_k} F(n; \{(s_i, e_i)\}_k, \{\theta_i\}_k, \theta_0, k) \]
which can be minimised using OP or PELT, and then use gradient descent for the outer minimisation over $\theta_0$ \citep{aPELT2019}. The main drawback of this approach is an increase in computation time proportional to the number of steps needed for the optimisation.

\section{Detection of changepoints with unknown background parameter}
\label{sectshort}

To estimate the background parameter and changepoints efficiently, while allowing a large proportion of non-background data, we introduce Algorithm \ref{algshort}. The algorithm makes a pass over the data, during which epidemic segments are detected in a standard way, i.e., by minimising penalized cost using OP (steps 3-12). The estimate of parameter $\theta_0$ is iteratively updated at each time point. We will show that this process converges almost surely to the true $\theta_0$ for certain models of background and segment levels, or reaches its neighbourhood $(\theta_0-\epsilon; \theta_0+\epsilon)$ under more general conditions.

\begin{algorithm}
\caption{Detection of changepoints with unknown background parameter}
\label{algshort}
    \begin{algorithmic}[1]
        \STATE Input: $C^0, C, \beta, x_{0:n}$
        \STATE Initialize $F(0) = 0, B(0)=\{x_0\}, \theta_0 = x_0$
        \FOR{$t \in 1,\dots,n$}
        \STATE Calculate $F_B = F(t-1) + C^0(x_t; \theta_0)$, and $F_S = \min_{1 \le k \le l} F(t-k) + C(x_{t-k+1:t}) + \beta$
            \IF{$F_B < F_S$}
                \STATE Assign $B(t) = B(t-1) \cup x_t$, and recalculate $\theta_0$ from $B$
                \STATE $F(t) = F_B$
            \ELSE
            \STATE Assign $B(t) = B(t-k)$ (here $k$: $\argmin F_S$ in step 4)
                \STATE $F(t) = F_S$
            \ENDIF
        \ENDFOR
        \STATE Repeat steps 3-12, without updating $\theta_0$
        \STATE Output: $\theta_0$, changepoint positions
    \end{algorithmic}
\end{algorithm}

The algorithm then makes a second pass over the data (step 13), repeating the segmentation using the final estimate of $\theta_0$. The step is identical to a single application of a known-background changepoint detector, such as the one described in \citet{CAPA2018}. Its purpose is to update the changepoint positions that are close to the start of the data and so had been determined using less precise estimates of $\theta_0$. This simplifies the theoretical performance analysis, but an attractive option is to use this algorithm in an online manner, without this step. We evaluate the practical consequences of this omission later in Section \ref{sectexperiments}.

By segmenting and estimating the background simultaneously, this algorithm is computationally more efficient than dedicated optimisation over possible $\theta_0$ values as in \citet{aPELT2019}, while allowing a larger fraction of non-background points than using a robust estimator on unsegmented data, as in \citet{CAPA2018}. 

We now demonstrate that Algorithm \ref{algshort} will converge almost surely under certain constraints, and is consistent. 

\subsection{Convergence}

\begin{theorem}
    \label{thm:algconverges}
    Consider the problem of an epidemic change in mean, with data $x_{0:n}$ generated as in \eqref{eq:epidemmodel}. Assume the distribution $f_B(x)$ and marginal distribution of segment points $f_S(x) = \int f_S(x; \theta) f_\Theta(\theta) d\theta$ are symmetric and strongly unimodal, with unknown background mean $\theta_0$, and that data points within each segment are iid. Here, $f_\Theta$ is the probability density of parameter $\theta$ corresponding to each data point. Denote by $w_t$ the estimate of $\theta_0$ obtained by analysing $x_{0:t}$ by Algorithm \ref{algshort}. The sequence $\{w_t\}$ converges:
    \begin{enumerate}
        \item to $\theta_0$ almost surely, if $\mathbb{E}f_S = \theta_0$.
        \item to a neighbourhood $(\theta_0 - \epsilon, \theta_0 + \epsilon)$ almost surely, where $\epsilon \rightarrow 0$ as the number of background points $n$ between successive segments $n\rightarrow \infty$.
    \end{enumerate}
\end{theorem}

We refer the reader to Appendix~\ref{app:convergence} for the proof. It is based on a result by \citet{bottou98}, who established conditions under which an online minimisation algorithm almost surely converges to the optimum (including stochastic gradient descent as a special case). We show that these conditions are either satisfied directly by the updating process in our algorithm, or longer update cycles can be defined that will satisfy the conditions, in which case the convergence efficiency is determined by $n$.

\subsection{Consistency}
The changepoint model can be understood as a function over an interval that is sampled to produce $n$ observed points. As the sampling density increases, more accurate estimation of the number and locations of changepoints is expected; this property is formalised as consistency of the detector.
\citet{CAPA2018} showed that detectors based on minimising penalised cost are consistent in the Gaussian data case, and their result can be adapted to prove consistency of Algorithm \ref{algshort}. We will use the strengthened SIC-type penalty $\alpha \log(n)^{1+\delta}, \delta>0$ as presented in \citet{CAPA2018}, but a similar result can be obtained with $\delta=0$.

Additionally, the minimum signal-to-noise ratio at which segments can be detected is formalised as this assumption:
\begin{equation}
    \forall i, (e_i-s_i)\Delta_i > \log(n)^{1+\delta},
    \label{eq:snrassumpt}
\end{equation}
where $\Delta_i$ is a function of the relative strength of the parameter changes $|\mu_i - \mu_0|$ and $\sigma_k/\sigma_0 + \sigma_0/\sigma_k$.

\begin{theorem}
    \label{thm:consist}
    Let the data $x_{0:n}$ be generated from an epidemic changepoint model as in \eqref{eq:epidemmodel}, with $f_B$ and $f_S$ Gaussian, and the changing parameter $\theta$ is either its mean or variance (assume the other parameter is known). Further, assume \eqref{eq:snrassumpt} holds for $k$ changepoints. Analyse the data by Algorithm \ref{algshort} with penalty $\beta = \alpha \log(n)^{1+\delta}$, $\alpha,\delta>0$.
    The estimated number and position of changepoints will be consistent, i.e. $\forall \epsilon>0, n>B$:
    \begin{equation}
        P\left(\hat{k} = k, |\hat{s_i}-s_i| < \frac{A}{\Delta_i} \log(n)^{1+\delta}, |\hat{e_i}-e_i| < \frac{A}{\Delta_i} \log(n)^{1+\delta}, \forall 1\le i \le k \right) \ge 1 - Cn^{-\epsilon},
        \label{eq:consist}
    \end{equation}
    for some $A, B, C$ that do not increase with $n$.
\end{theorem}

The proof is given in Appendix~\ref{app:consistency}. We use the connection between our Algorithm \ref{algshort} and stochastic gradient descent to establish error bounds on the backgorund parameter estimates. These bounds then allow us to apply a previous consistency result \citep{CAPA2018} to our case.

\subsection{Pruning}

Much of the improvement in performance of change estimation algorithms comes from pruning of the search space. Standard pruning \citep{KillickPELT} can be applied to Algorithm \ref{algshort}. It continues to find optimal solutions in this case, as is shown in Proposition~\ref{prop:pruning}.

To implement pruning, the search set for $F_S$ in step 4 would be changed to:
\[ F_S = \min_{t' \in K}(F(t') + C(x_{t'+1:t}) + \beta), \]
and step 11a would be introduced to update $K$ as:
\begin{equation*}
    \text{11a. } K = K \cap \{ s : F(s) + C(x_{s+1:t}) < F(t),~ t+1-l \le s < t \} \cup \{t\}. \\
\end{equation*}

\begin{prop}
\label{prop:pruning}
    Assume a cost function $C$ such that \eqref{eq:pruning} applies. If, for some $s > t-l$:
    \[F(s) + C(x_{s+1:t}) \ge F(t), \]
    then $t-k$ will not be a segment start point in the optimal solution, and can be excluded from consideration for all subsequent $t'>t$.
\end{prop}
\begin{proof}
    For all $t'\le s+l$, the proof applies as for other pruned algorithms \citep{KillickPELT}:
    \[ F(t) + C(x_{t+1:t'}) + \beta \le F(s) + C(x_{s+1:t}) + C(x_{t+1:t'}) + \beta \le F(s) + C(x_{s+1:t'}) + \beta. \]
    For $t' > s+l$, segment $(s, t')$ will exceed the length constraint and thus cannot be part of the final segmentation.
\end{proof}

\section{Detecting changepoints with a nuisance process}
\label{sectmethod}

\subsection{Problem setup}
In this section, we consider the changepoint detection problem when there is an interfering nuisance process. We assume that this process, like the signal, consists of segments, and we index the start and end times as $s^N_j, e^N_j$. Data within these segments is generated from a nuisance-only distribution $f_N$, or from some distribution $f_{NS}$ if a signal occurs at the same time. In total, four states are possible (background, nuisance, signal, or nuisance and signal), so the overall distribution of data points is:
\begin{equation}
    f(x_t) = \begin{cases} f_{NS}(x_t; \theta_i, \theta^N_j) & \text{ if } \exists i,j: t \in [s_i; e_i] \cap [s^N_j, e^N_j] \\
        f_S(x_t; \theta_i) & \text{ if } \exists i: t \in [s_i, e_i], t \notin \cup_j [s^N_j, e^N_j] \\
        f_N(x_t; \theta^N_j) & \text{ if } \exists j: t \in [s^N_j, e^N_j], t \notin \cup_i [s_i, e_i] \\
        f_B(x_t) & \text{ otherwise} \end{cases}
    \label{eq:nuisancemodel}
\end{equation}

We add two more conditions to ensure identifiability: 1) the nuisance process evolves more slowly than the signal process, so $\min(e^N_j-s^N_j) > \max(e_i-s_i)$; 2) signal segments are either entirely contained within a nuisance segment, or entirely out of it:
\begin{equation}
    \forall i,j, \text{ either } [s_i, e_i] \subset (s^N_j, e^N_j), \text{ or } [s_i, e_i] \cap [s^N_j, e^N_j] = \emptyset
    \label{eq:condnooverlap}
\end{equation}

We adapt the changepoint detection method as follows. Define $X_S = \bigcup_i x_{s_i:e_i}, X_N = \bigcup_j x_{s^N_j:e^N_j}$, a penalty $p'(m) = \beta'm$ for the number of nuisance segments, and cost functions $C^{NS},~C^S,~C^N,~C^0$ corresponding to each of the distributions in \eqref{eq:nuisancemodel}. Then the full cost of a model with $k$ true segments and $m$ nuisance segments is:
\begin{multline}
    \label{eq:fullcostnuis}
    F(n; \{(s_i, e_i)\}_k, k, \{(s^N_j, e^N_j)\}_m, m, \bm{\theta}) = C^0(x_{0:n} \setminus (X_S \cup X_N) ) + \sum_{i=0}^k C^S(x_{s_i:e_i} \setminus X_N; \theta_i) + \\
    + \sum_{j=0}^m \left( C^N(x_{s^N_j:e^N_j} \setminus X_S; \theta_j) + \sum_{i=0}^k C^{NS}(x_{s^N_j:e^N_j} \cap x_{s_i:e_i}; \theta_i, \theta_j) \right) + \beta k + \beta' m.
\end{multline}

Note that this cost can also be expressed, using $k_j$ as the number of signal segments that overlap a nuisance segment $j$ and $k_0 = k - \sum k_j$, as:
\begin{multline*}
    F(n; \{(s_i, e_i)\}_k, k, \{(s^N_j, e^N_j)\}_m, m, \bm{\theta})
    = C^0(x_{0:n} \setminus (X_S \cup X_N) ) + \sum_{i=0}^{k_0} C^S(x_{s_i:e_i} \setminus X_N; \theta_i) + \\
     + \beta k_0 + \sum_{j=0}^m (C'(x_{s^N_j:e^N_j}) + \beta k_j + \beta').
\end{multline*}
Condition \eqref{eq:condnooverlap} ensures that $C'$ and $C^S$ are independent (no points or $\theta_i$ are shared between them), so $F$ is block-additive over (signal or nuisance) segments and can be minimised by dynamic programming.

\subsection{Proposed method}
To minimise this cost, we propose the method outlined in Algorithm \ref{algfull}.
Its outer loop proceeds over the data to identify segments by the usual OP approach. However, to evaluate the cost $C'(x_{a:b})$, which allows segment and nuisance overlaps, a changepoint detection problem with unknown background level must be solved over $x_{a:b}$. This is achieved by an inner loop using Algorithm \ref{algshort}. If $C'=C^S$, this method would reduce to a standard detector of epidemic changepoints, with the exception that segments are separated into two types depending on their length; this is very similar to the collective and point anomaly (CAPA) detector in \citet{CAPA2018}.

The cost $C'$ is minimised using Algorithm \ref{algshort}, so by Theorem \ref{thm:consist}, the number of positions of true segments will be estimated consistently given accurate assignment of the $m$ nuisance positions $s^N_j, e^N_j$. However, the latter event is subject to a complex set of assumptions on relative signal strength, position, and duration of the segments. Therefore, we do not attempt to describe these in full here, but instead investigate the performance of the method by extensive simulations in Section \ref{simmethod}.

\begin{algorithm}
    \caption{Adaptive segmentation by changepoint detection}
\label{algfull}
\begin{algorithmic}[1]
    \STATE Input: $C^0, C^S, C^N, C^{NS}, \beta, \beta', x_{0:n}$
    \STATE Initialize $F(0) = 0, \theta_0 = x_0$, lists of tuples $chp_S, chp_N, chp_{NS}$
    \STATE If $\theta_0$ not known: estimate using an appropriate robust estimator
    \FOR{$t \in 1,\dots,n$}
        \FOR{$t' \in 0,\dots,t-l-1$}
            \STATE Apply Algorithm \ref{algshort} to data $x_{t'+1:t}$, store returned cost as $C'$ and changepoints as $chp_{NS}(t')$
            \STATE $F_N(t') = F(t') + C' + \beta'$
        \ENDFOR
        \STATE $F_B = F(t-1) + C^0(x_t)$
        \STATE $F_S = \min_{1 \le k \le l} F(t-k) + C^S(x_{t-k+1:t}) + \beta$
        \STATE $F_N = \min F_N(t')$
        \STATE Assign $F(t) = \min\{F_B, F_S, F_N\}$
        \IF{$F(t) = F_B$}
            \STATE Assign $chp_S(t) = chp_S(t-1), chp_N(t)=chp_N(t-1)$
        \ELSIF{$F(t) = F_S$}
        \STATE Assign $chp_S(t) = chp_S(t-k) \cup (t-k, t)$ with $k$ determined by $\argmin F_S$, $chp_N(t)=chp_N(t-1)$
        \ELSE
        \STATE Assign $chp_S(t) = chp_S(t') \cup chp_{NS}(t')$, $chp_N(t) = chp_N(t') \cup (t', t)$ with $t'$ determined by $\argmin F_N$ 
        \ENDIF
    \ENDFOR
    \STATE Output: changepoint positions $chp_S(n) = \{(s_i, e_i)\}$ and $chp_N(n) = \{(s^N_j, e^N_j)\}$
\end{algorithmic}
\end{algorithm}

Algorithm \ref{algfull} is stated assuming that a known or estimated value of the parameter $\theta_0$, corresponding to the background level without the nuisance variations, is available. In practice, it may be known when there is a technical noise floor or a meaningful baseline that can be expected after removing the nuisance changes. Alternatively, $\theta_0$ may be estimated by a robust estimator if a sufficient fraction of background points can be assumed; in either case, the (known or estimated) value is substituted to reduce the computational cost. Otherwise, the method can be modified to estimate $\theta_0$ simultaneously with segmentation, using a principle similar to Algorithm \ref{algshort}.

\subsection{Pruning}

In the proposed method, the estimation of parameter $\mu_j$ (the mean of segment $j$) is sensitive to the segment length, therefore the cost $C'$ does not necessarily satisfy the additivity condition \eqref{eq:pruning}, and so it cannot be guaranteed that PELT-like pruning will be exact. However, we can establish a local pruning scheme that retains the exact optimum with probability $\rightarrow 1$ as $n \rightarrow \infty$.

\begin{prop}
    \label{prop:newpruning}
    Assume data $x_{0:n}$ is generated from a Gaussian epidemic changepoint model, and that the distance between changepoints is bounded by some function $A(n)$:
    \[ \forall i,j: s^N_{j+1} - e^N_j > A(n), |s_i - s^N_j| > A(n). \]

    At time $t$, the solution space is pruned by removing:
    \begin{equation}
        \label{eq:kpruned}
         \bm{k}_{pr, t} = \{k : F(k-1) + C'(x_{k:t}) \ge \min_m F(m-1) + C'(x_{m:t}) + \alpha \log(n)^{1+\delta} \}.
    \end{equation}
    Here $m\in (t-A(n); t], k \in (t-A(n); t], k \neq m$.
    Then $\forall \epsilon>0$, there exist constants $B, n_0$, such that when $n>n_0$, the true nuisance segment positions are retained with high probability:
    \[ P( \forall j: s^N_j \notin \bigcup_t \bm{k}_{pr, t} , e^N_j \notin \bigcup_t \bm{k}_{pr, t}) \ge 1-Bn^{-\epsilon}. \]
\end{prop}

The proof is given in Appendix~\ref{app:pruning}.
The assumed distance bound serves to simplify the detection problem: within each window $(t-A(n), t]$, at most 1 true changepoint may be present, and the initial part of Algorithm \ref{algfull} is identical to a standard epidemic changepoint detector. It can be shown that other candidate segmentations in the pruning window are unlikely to have significantly lower cost than the one associated with $s^N_j, e^N_j$, and therefore $s^N_j, e^N_j$ are likely to be retained in pruning.
    
As this scheme only prunes within windows of size $A(n)$, and setting large $A(n)$ may violate the true data generation model, it is less efficient than PELT-type pruning. However, assuming that the overall estimation of nuisance and signal changepoints is consistent, Proposition \ref{prop:newpruning} extends to standard pruning over the full dataset.
We show that this holds empirically in Section \ref{simmethod}.

\section{Experiments}
\label{sectexperiments}

In this section, we present the results of simulations used to evaluate the performance of the methods, and demonstrations on real-world datasets.

The methods proposed in this paper are implemented in R 3.5.3, and available on GitHub at \url{https://github.com/jjuod/changepoint-detection}. The repository also includes the code for recreating the simulations and data analyses presented here.

\subsection{Simulations}

\subsubsection{Algorithm \ref{algshort} estimates the background parameter consistently}
\label{simulationsetup1}
In the first experiment we tested the performance of the algorithm with a fixed background parameter.
Three different sets of data were generated and used:  

\textbf{Scenario 1.} Gaussian data with one signal segment: $n$ data points were drawn as $x_t \sim \mathcal{N}(\theta_t, 1)$, with a segment of $\theta_t = 3$ at times $t \in (0.3n; 0.5n]$, and $0$ otherwise (background level).  

\textbf{Scenario 2.} Gaussian data with multiple signal segments: $n$ data points were drawn as $x_t \sim \mathcal{N}(\theta_t, 1)$, with:
\begin{equation*}
    \theta_t = \begin{cases} -1 \text{ when } t \in (0.2n; 0.3n] \cup (0.7n; 0.8n] \\
        1 \text{ when } t \in (0.5n; 0.6n] \\
        0 \text{ otherwise }.
    \end{cases}
\end{equation*}

\textbf{Scenario 3.} Heavy tailed data with one signal segment: $n$ data points were drawn from the generalized $t$ distribution as $x_t \sim T(3) + \theta_t$, with $\theta_t = 2$ at times $t \in (0.2n; 0.6n]$ and $0$ otherwise (background level).

To evaluate the background level ($\theta_0$) estimation by Algorithm \ref{algshort}, we generated time series for each of the three scenarios with values of $n$ between 30 and 750, with 500 replications for each $n$. Note that the segment positions remain the same for all sample sizes, so the increase in $n$ could be interpreted as denser sampling of the time series.

Each time series was analysed using Algorithm \ref{algshort} to estimate $\theta_0$. The maximum segment length was set to $l=0.5n$ and the penalty to $\beta=3 \log(n)^{1.1}$. The cost was computed using the Gaussian log-likelihood function with known variance. This means that the cost function was mis-specified in Scenario 3, and so provides a test of the robustness of the algorithm (although the variance was set to $3$, as expected for $T(3)$).

For comparison, in each replication we also retrieved the median of the entire time series $x_{1:n}$, which is used as the estimate of $\theta_0$ by \citet{CAPA2018}. For scenarios (1) and (2), we also computed the quantiles of $\mathcal{N}(0, \sigma^2/\sqrt{n_B})$ as the oracle efficiency limit based on the CLT, with $n_B$ the total number of background points: $n_B=0.8n$ for (1) and $0.7n$ for (2).

Figure~\ref{fig:sim1theta} shows the results for all three scenarios. It can be seen that Algorithm \ref{algshort} produced consistent estimates given sufficiently long time series. When the signal-to-noise ratio was large (scenario 1), most of the estimated background values were accurate even at small $n$, and at larger $n$ our estimator reached efficiency and accuracy similar to the oracle estimator. Predictably, the median of the non-segmented data produced biased estimates, which nonetheless may be preferrable in some cases, as our estimate showed large variability at low $n$ in scenarios 2 and 3. At $n>400$, our estimator provides the lowest total error in all tested scenarios, even with the mis-specified cost function in scenario 3.

\begin{figure}[h]
    \centering
    \includegraphics[width=0.9\textwidth]{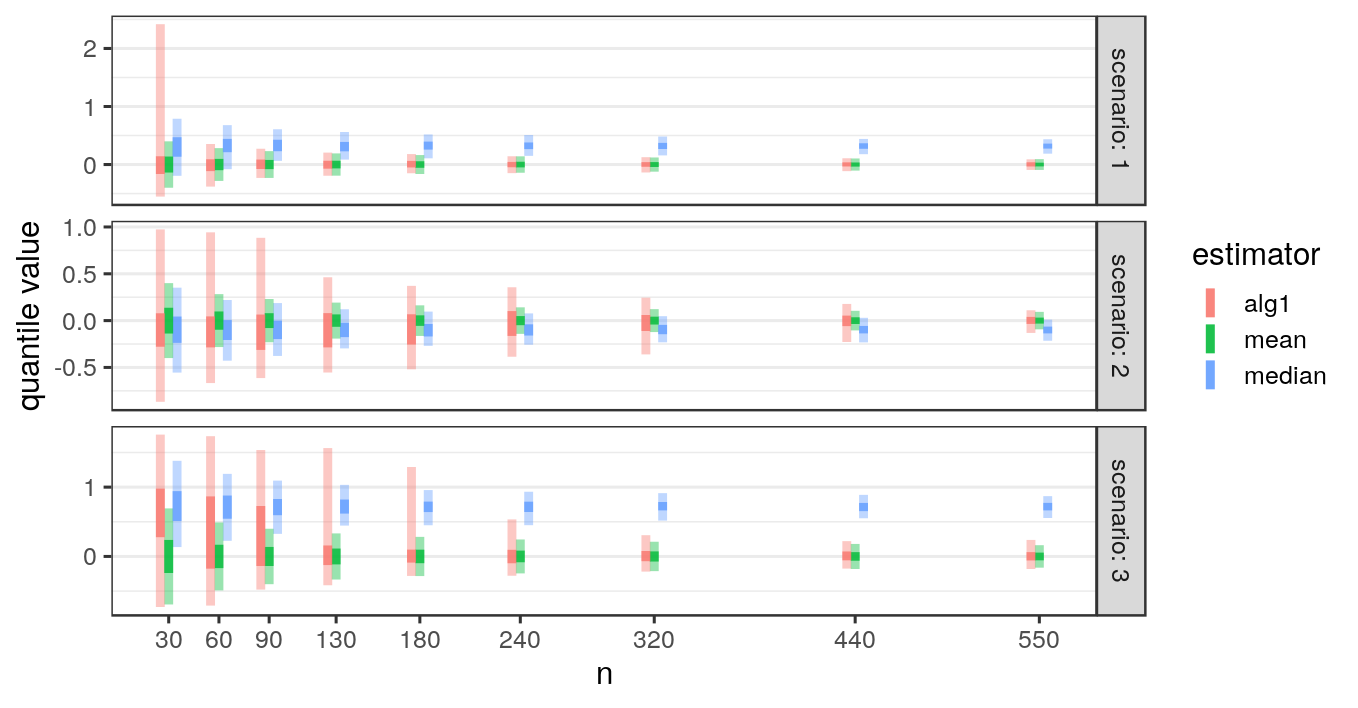}
    \caption{Consistency of the background parameter values estimations. Time series simulated in three different scenarios were analysed by Algorithm \ref{algshort} (shown in red). Lines are the inter-quartile range (solid) and 2.5-97.5\% range (faint) of the background parameter estimates observed in 500 replications at each $n$. For comparison, we also show the ranges of estimates obtained by the median (blue) and mean of true background points (green; an oracle estimator).}
    \label{fig:sim1theta}
\end{figure}

\subsubsection{Segment positions estimated by Algorithm \ref{algshort} are accurate and consistent}

The same simulation setup was used to evaluate the consistency of estimated segment number and positions. Data were generated and analysed with Algorithm \ref{algshort} as in Section \ref{simulationsetup1}, but we also retrieved the changepoint positions that were estimated in step 12. This corresponds to the online usage of the algorithm, in which segmentation close to the start of the data is based on early $\theta_0$ estimates and may therefore itself be less precise.
We extracted the mean number of segments reported by the algorithm and also calculated the true positive rate (TPR) as the fraction of simulations in which the algorithm reported at least 1 changepoint within $0.05n$ points of each true changepoint.

The results show that the TPR approaches 1 for all three scenarios (Table \ref{tab:sim1segs}). As expected, detecting a strong segment was easier than multiple weak ones: segmentation in scenario 1 was accurate even at $n=30$, while $n > 500$ was needed to reach $\ge 90\%$ TPR in scenario 2. In scenario 3, the algorithm correctly detected changes at the true segment start and end (TPR $\approx 100 \%$), but the algorithm tended to fit the segment as multiple ones, likely due to the heavy tails of the $t$ distribution. Notably, skipping step 13 had very little impact on the performance of the algorithm, suggesting that the faster online version can be safely used.

\begin{table}
    \caption{\label{tab:sim1segs} Consistency of the number and position of estimated changepoints. Time series simulated in three different scenarios were analysed using Algorithm \ref{algshort}. For each $n$, 500 replications were performed. The mean number of reported segments and the TPR (fraction of iterations when a changepoint was detected within $0.05n$ of each true changepoint) are shown. The number of true segments was 1, 3, and 1 for scenarios 1, 2, 3 respectively. The same time series were also analysed in an online manner, i.e., without Step 13 of the algorithm; the results are shown on the right of the table.}
    \centering
    \begin{tabular}[htbp]{lcrrrr}
        \toprule
        & & \multicolumn{2}{c}{(Full algorithm)}  & \multicolumn{2}{c}{(no step 13)} \\
        Scenario & $n$ & Mean \# segm. & TPR & Mean \# segm. & TPR \\  \midrule
        1  &       30 &      1.112 &    0.932 &      1.138 &    0.924 \\
        &       90 &         1.040 &    0.998 &      1.056 &    0.998 \\
        &      180 &         1.054 &    1.000 &      1.060 &    1.000 \\
        &      440 &         1.032 &    1.000 &      1.040 &    1.000 \\
        &      750 &         1.028 &    1.000 &      1.034 &    1.000 \\ \midrule
        2  &       30 &      0.552 &    0.000 &      0.584 &    0.000 \\
         &       90 &        1.096 &    0.008 &      1.050 &    0.004 \\
         &      180 &        1.762 &    0.086 &      1.720 &    0.080 \\
         &      440 &        2.900 &    0.824 &      2.854 &    0.782 \\
         &      750 &        3.010 &    0.992 &      3.008 &    0.988 \\ \midrule
        3  &       30 &      0.644 &    0.142 &      0.662 &    0.132 \\
         &       90 &        1.426 &    0.592 &      1.416 &    0.570 \\
         &      180 &        1.930 &    0.878 &      1.932 &    0.866 \\
         &      440 &        2.798 &    0.998 &      2.798 &    0.996 \\
         &      750 &        3.692 &    1.000 &      3.664 &    1.000 \\ \bottomrule
    \end{tabular}
\end{table}

\subsubsection{Algorithm \ref{algfull} recovers true signal segments under interference}
\label{simmethod}

To evaluate Algorithm \ref{algfull}, we generated time series with both signal and nuisance segments under two different scenarios:

\textbf{Scenario 1.} Gaussian data with a signal segment overlapping a nuisance segment: $n$ data points were drawn as $x_t \sim \mathcal{N}(\theta^S_t + \theta^N_t, 1)$, with $\theta^S_t = 2$ when $t \in (0.3n; 0.5n]$, 0 otherwise, and $\theta^N_t = 2$ when $t \in (0.2n; 0.7n]$, 0 otherwise.

\textbf{Scenario 2.} Gaussian data with a nuisance segment and two non-overlapping signal segments: $n$ data points were drawn as $x_t \sim \mathcal{N}(\theta^S_t + \theta^N_t, 1)$, with:
\begin{align*}
    \theta^N_t = 1 \text{ when } t \in (0.2n; 0.4n],~0 \text{ otherwise} \\
    \theta^S_t = \begin{cases} 3 \text{ when } t \in (0.5n; 0.6n] \\
        -3 \text{ when } t \in (0.7n; 0.8n] \\
        0 \text{ otherwise }
    \end{cases}
\end{align*}

Time series were generated for $n$ between 30 and 240, in 500 replications at each $n$.
Each series was analysed by three methods. Algorithm \ref{algfull} (\textbf{proposed}) was run with penalties $\beta=\beta'=3 \log(n)^{1.1}$ as before. For classical detection of epidemic changes in mean, we used the R package \textbf{anomaly} \citep{CAPA2018}; the implementation allows separating segments of length 1, but we treated them as standard segments and set all penalties to $3 \log(n)^{1.1}$. In both methods, background parameters were set to $\mu_0=0, \sigma_0=1$, and maximum segment length was $l=0.33n$ in scenario 1 and $l=0.15n$ in scenario 2. As an example of a different approach, we included the narrowest-over-threshold detector implemented in R package \textbf{not}, with default parameters. This is a non-epidemic changepoint detector that was shown to outperform most comparable methods \citep{Baranowski2019}. Since it does not include the background-signal distinction, we define signal segments as regions between two successive changepoints where the mean exceeds $\mu_0 \pm \sigma_0$.

As before, evaluation metrics were the mean number of segments and the TPR. For the proposed method, we required accurate segment type detection as well, i.e., a true positive is counted when the detector reported a nuisance start/end within $0.05n$ of each nuisance changepoint and a signal start/end within $0.05n$ of each signal changepoint. In scenario 1, we also extracted the estimate of $\theta$ corresponding to the detected signal segment (or segment closest to $(0.3n; 0.5n]$ if multiple were detected). The average of these over the 500 replications is reported as $\hat{\theta}^S$.

To evaluate the effects of pruning, Algorithm \ref{algfull} was applied without pruning, or with global pruning as in \eqref{eq:kpruned}, with $m\in (0; t-l)$ at each $t$. Only 5 out of 5000 runs (500 iterations $\times$ 10 settings) showed any differences between the globally-pruned and non-pruned methods (Supplementary Table S1 in Appendix \ref{app:suppfigs}), so we only present results obtained with pruning from here on.

We observed that the proposed method (with pruning) successfully detected true signal segments in both scenarios (Figure \ref{fig:sim2bias} and Table \ref{tab:sim2tprs}). The number of nuisance detections was accurate in scenario 1, and slightly underestimated in favour of more signal segments in scenario 2, most likely because the simulated nuisance length was close to the cutoff of $0.15n$.

\begin{figure}[h]
    \centering
    \includegraphics[width=0.7\textwidth]{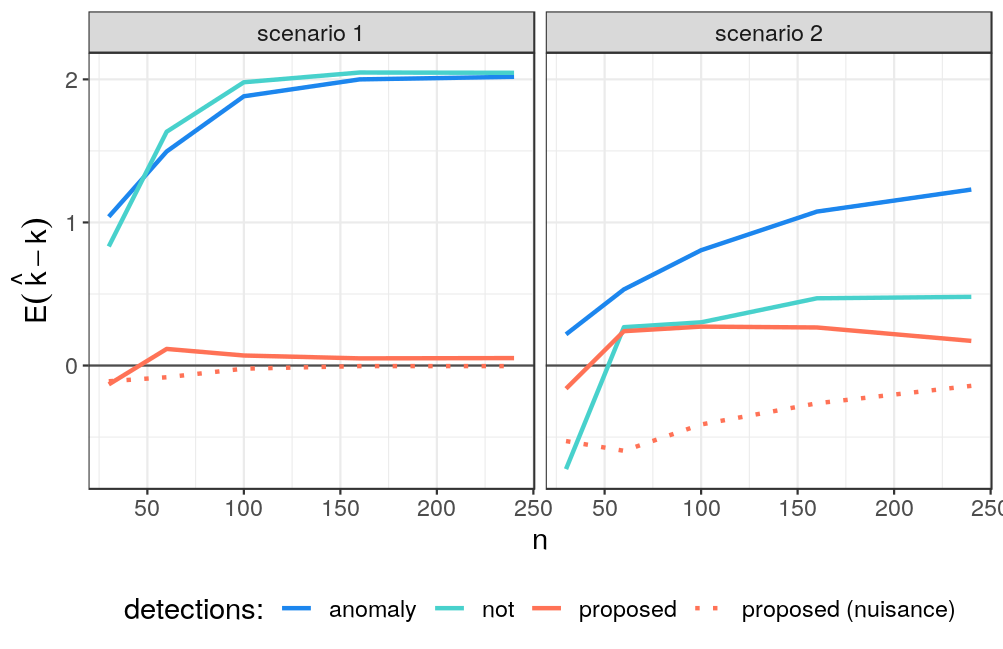}
    \caption{Relative bias in the number of changepoints estimated by the proposed Algorithm \ref{algfull} (pruned), and two alternative detectors: \textbf{anomaly} and \textbf{not}. Data simulated for two scenarios, in 500 replications for each $n$. For the proposed algorithm, bias is calculated separately in signal and nuisance segments.}
    \label{fig:sim2bias}
\end{figure}

\begin{table}
    \caption{\label{tab:sim2tprs} The true positive rate of changepoint estimation by the proposed Algorithm \ref{algfull} (pruned), and two alternative detectors: \textbf{anomaly} and \textbf{not}. Data simulated for two scenarios, in 500 replications for each $n$. The true positive rate is the fraction of iterations when a changepoint of correct type was detected within $0.05n$ of each true changepoint (types are signal, S, and nuisance, N).}
    \centering
    \begin{tabular}[htbp]{lccccc}
        \toprule
        Scenario & $n$ & Proposed (S) & Proposed (N) &  anomaly  &   not \\ \midrule
        1  &       30 &  0.444  &     0.548   &    0.382   &  0.428  \\
           &       60 &  0.742  &     0.752   &    0.532   &  0.724  \\
           &      100 &  0.938  &     0.930   &    0.870   &  0.912  \\
           &      160 &  0.984  &     0.972   &    0.982   &  0.994  \\
           &      240 &  1.000  &     0.986   &    1.000   &  0.998  \\ \midrule
        2  &       30 &  0.700  &     0.110   &    0.934   &  0.152  \\
           &       60 &  0.924  &     0.238   &    0.962   &  0.872  \\
           &      100 &  0.986  &     0.412   &    0.998   &  0.988  \\
           &      160 &  0.998  &     0.640   &    1.000   &  1.000  \\
           &      240 &  0.998  &     0.764   &    1.000   &  1.000  \\ \bottomrule
    \end{tabular}
\end{table}

As expected, when nuisance segments are not included in the model (\textbf{anomaly} and \textbf{not} methods), they were identified as multiple changepoints; as a result, the number of segments was over-estimated up to 3-fold. These models are also unable to capture the signal-specific change in mean $\theta^S$: \textbf{anomaly} estimated $\hat{\theta}^S=3.99$, and \textbf{not} estimated $\hat{\theta}^S=3.95$ in scenario 1 at $n=240$. These values correspond to the sum of the signal and nuisance effects. While the estimation is accurate and could be used to recover the value of interest by post-hoc analysis, our proposed method estimated the signal-specific change directly, as $\hat{\theta}^S=2.00$ in scenario 1 at $n=240$.


\subsection{Real-world Data}
\subsubsection{ChIP-seq}

As an example application of the algorithms proposed in this paper, we demonstrate peak detection in chromatin immunoprecipitation sequencing (ChIP-seq) data. The goal of ChIP-seq is to identify DNA locations where a particular protein of interest binds, by precipitating and sequencing bound DNA. This produces a density of binding events along the genome that then needs to be processed to identify discrete peaks.
Typically, some knowledge of the expected peak length is available to the researcher, although with considerable uncertainty \citep{ZINBA2011}. Furthermore, the background level may contain local shifts of various sizes, caused by sequencing bias or true structural variation in the genome \citep{MACS2008}.
The method proposed in this paper is designed for such cases, and can potentially provide more accurate and more robust detection.

We used datasets from two ChIP-seq experiments, investigating histone modifications in human immune cells.
Broad Institute H3K27ac data was obtained from UCSC Genome Browser, GEO accession GSM733771, as mean read coverage in non-overlapping windows of 25 bp. While ground truth is not available for this experiment, we also retrieved an input control track for the same cell line from UCSC (GEO accession GSM733742).
We analysed a window near the centromere of chromosome 1, between 120,100,000 to 120,700,000 bp. This window contains nuisance variation in the background level, as seen in the input control (Figure \ref{fig:chipbroad}, top). To improve runtime, data was downsampled to approximately 1000 points, each corresponding to mean coverage in a window of 500 bp.
All read coordinates in this paper correspond to hg19 genome build.

The second dataset was the UCI/McGill dataset, obtained from \url{https://archive.ics.uci.edu/ml/datasets/chipseq}. This dataset was previously used for evaluating peak detection algorithms \citep{HockingBio2016, Hocking2018}. Mean read coverage at 1 bp resolution is provided, as well as peak annotations based on visual inspection. The annotations are weak labels in the sense that they indicate peak presence or absence in a region, not their exact positions, to acknowledge the uncertainty when labelling visually. 
From this data, we used the H3K36me3 modification in monocyte sample ID McGill0104, AM annotation. Around the labels $L=\{(s_i, e_i)\}$ in each chromosome, we extracted read coverage for the window between $s-(e-s)$ and $e+(e-s)$ bp, $s=\min s_i$, $e=\max e_i$, and downsampled to about 1000 points as before. Based on visual inspection of the coverage, we chose a group of labels in chromosome 12, which provides a variety of annotations and a structure that appears to contain both nuisance shifts and signal peaks.

The ChIP-seq datasets were analysed by the method proposed here (Algorithm \ref{algfull}), an epidemic detector from package \textbf{anomaly}, and the non-epidemic detector \textbf{not}. The length of the signal segments was limited to 50 kbp in \textbf{anomaly} and the proposed method. As an estimate of global mean $\mu_0$, the median of the unsegmented data was used, and $\sigma_0$ estimated by the standard deviation of non-segmented data. As before, penalties were set to $3 \log(n)^{1.1}$. Only segments with estimated $\theta>\mu_0$ are shown, as we are a priori interested only in regions of increased binding.

We also used GFPOP, implemented in R package \textbf{PeakSegDisk} \citep{Hocking2018}: this detector has been developed specifically for ChIP-seq data processing, and models a change in Poisson rate parameter. This method does not include a single background level, but enforces alternating up-down constraints. It is intended as a supervised method, with the penalty value $\lambda$ chosen based on the best segmentation provided by the training data. Therefore, for the Broad Institute dataset, we repeated the segmentation with $\lambda \in \{10^1, 10^2, 10^3, 10^4, 10^5, 10^6\}$, and show the $\lambda$ that produces between 2 and 10 segments.

To evaluate the results quantitatively, we calculated the SIC based on each method's segmentation. We used Gaussian likelihood with parameters matching the detection status (i.e., estimated mean for the points in each segment $\hat{\mu}_0$ for points outside segments, and $\hat{\sigma}_0$ as the standard deviation for all points). The number of parameters was set to $3k+2$: a mean and two endpoints for each of the $k$ segments reported, and two for the background parameters.

In the H3K27ac data, all methods detected the three most prominent peaks, but produced different results for smaller peaks and more diffuse change areas (Figure \ref{fig:chipbroad}, bottom).
Both \textbf{PeakSegDisk} and \textbf{anomaly} marked a broad segment in the area around 120,600,000 bp. Based on comparison with the control data, this change is spurious, and it exceeds the 50 kbp bound for target segments. While this bound was provided to the \textbf{anomaly} detector, it does not include an alternative way to model these changes, and therefore still reports one or more shorter segments.
In contrast, our method accurately modelled the area as a nuisance segment with two overlapping sharp peaks.

Using \textbf{not}, the data was partitioned into 10 segments. By defining segments with low mean ($\theta < \hat{\mu}_0 + \hat{\sigma}_0$) as background, we could reduce this to 4 signal segments; while this removed the spurious background change, it also discarded the shorter change around 120,200,000 bp, which fits the definition of a signal peak ($<50$ kbp) and was retained by the proposed method.
This data illustrates that choosing the post-processing required for most approaches is not trivial, and can have a large impact on the results. In contrast, the parameters required for our method have a natural interpretation and may be known a priori or easily estimated, and the outputs are provided in a directly useful form.
This data illustrates that choosing the post-processing options is not trivial, and can have a large impact on the results. In contrast, the parameters required for our method have a natural interpretation and may be known a priori or easily estimated, and the outputs are provided in a directly useful form.

\begin{figure}[h]
    \centering
    \includegraphics[width=0.9\textwidth]{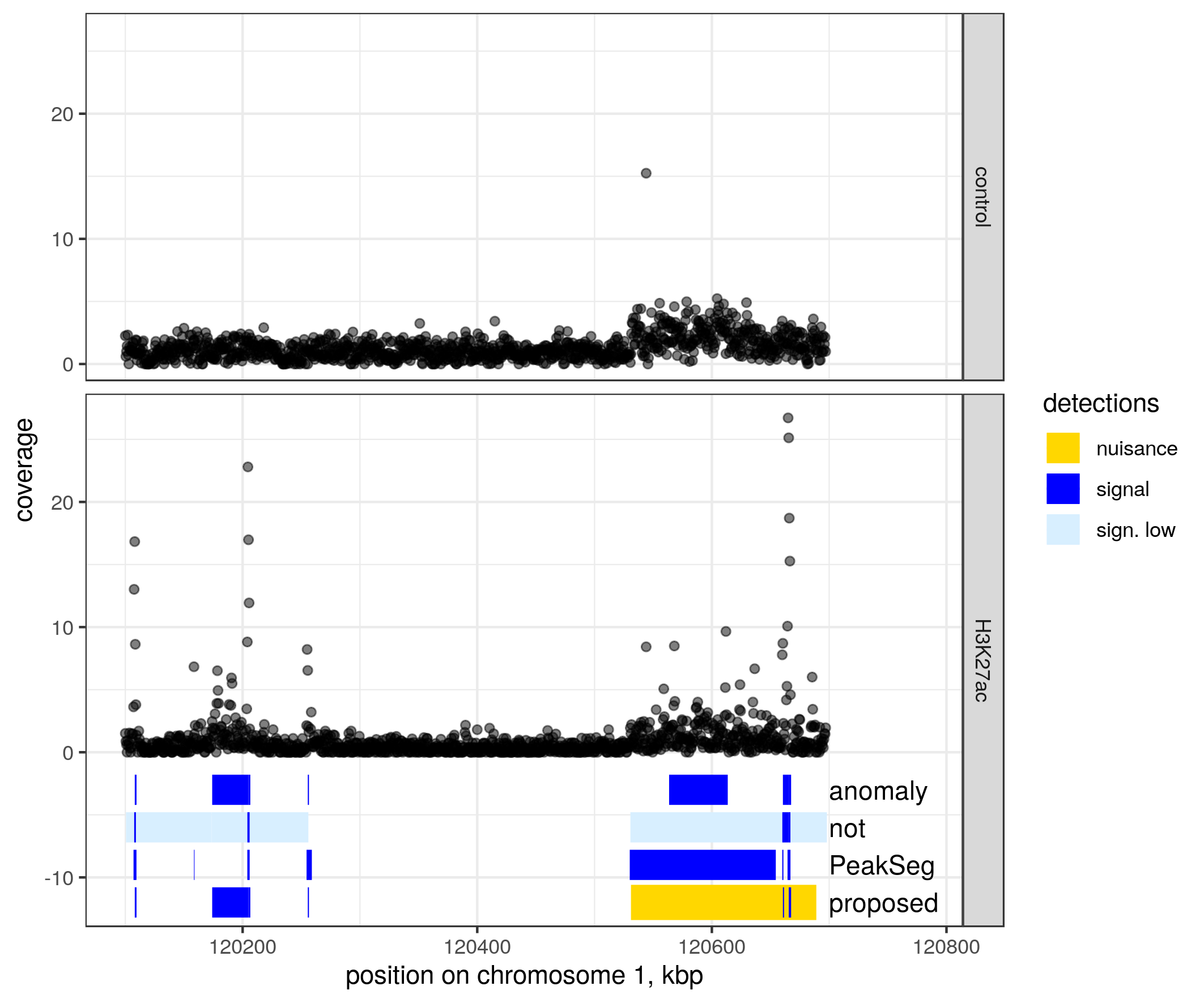}
    \caption{ChIP-seq read counts and analysis results. Counts provided as mean coverage in 500 bp windows for a non-specific control sample (top) and H3K27ac histone modification (bottom), chromosome 1. Segments detected in the H3K27ac data by the method proposed here (Algorithm \ref{algfull}) and three other detectors are shown under the counts. Note that the proposed method can also produce longer nuisance changes (yellow) overlapped by signal segments (blue). \textbf{not} does not specifically identify background segments; we show the ones with relatively low mean in light blue.}
    \label{fig:chipbroad}
\end{figure}

In the UCI data, segment detections also generally matched the visually determined labels. However, our method produced the most parsimonious models to explain the changes. Figure \ref{fig:ucichr12} shows such an example from chromosome 12, where our method reported two nuisance segments and a single sharp peak around 62,750,000 bp. The nuisance segments correspond to broad regions of mean shift, which were also detected by \textbf{anomaly} and \textbf{not}, but using 6 and 16 segments, respectively. Notably, \textbf{PeakSeg} differed considerably: as this method does not incorporate a single background level, but requires segments to alternate between background and signal, the area around 62,750,000 bp was defined as background, despite having a mean of $4.5~\hat{\mu}_0$. In total, 12 segments were reported by this method. This shows that the ability to separate nuisance and signal segments helps produce more parsimonious models, and in this way minimises the downstream efforts such as experimental replication of the peaks.

The visual annotations provided for this region are shown in the first row in Figure \ref{fig:ucichr12}. Note that they do not distinguish between narrow and broad peaks (single annotations in this sample range up to 690 kb in size). Furthermore, comparison with such labels does not account for finer segmentation, coverage in the peak area, or the number of false alarms outside it. For these reasons we are unable to use the labels in a quantitative way.

Quantitative comparison of the segmentations by SIC also favours our proposed method in both datasets. In the Broad dataset, SICs corresponding to segmentations reported by \textbf{PeakSeg}, \textbf{not} and \textbf{anomaly} were 4447.4, 4532.2, and 4311.6, respectively, while the segmentation produced by our model had an SIC of 4285.8. 
The smallest criterion values in UCI data were also produced by our method (4664.5), closely followed by \textbf{anomaly} (4664.8), while \textbf{not} segmentation resulted in an SIC of 4705.7 and \textbf{PeakSeg} 6886.7.
This suggests that in addition to the practical benefits of separating unwanted segments, the nuisance-signal structure provides a better fit to these datasets than models that allow only one type of segments.

\begin{figure}[h]
    \centering
    \includegraphics[width=0.9\textwidth]{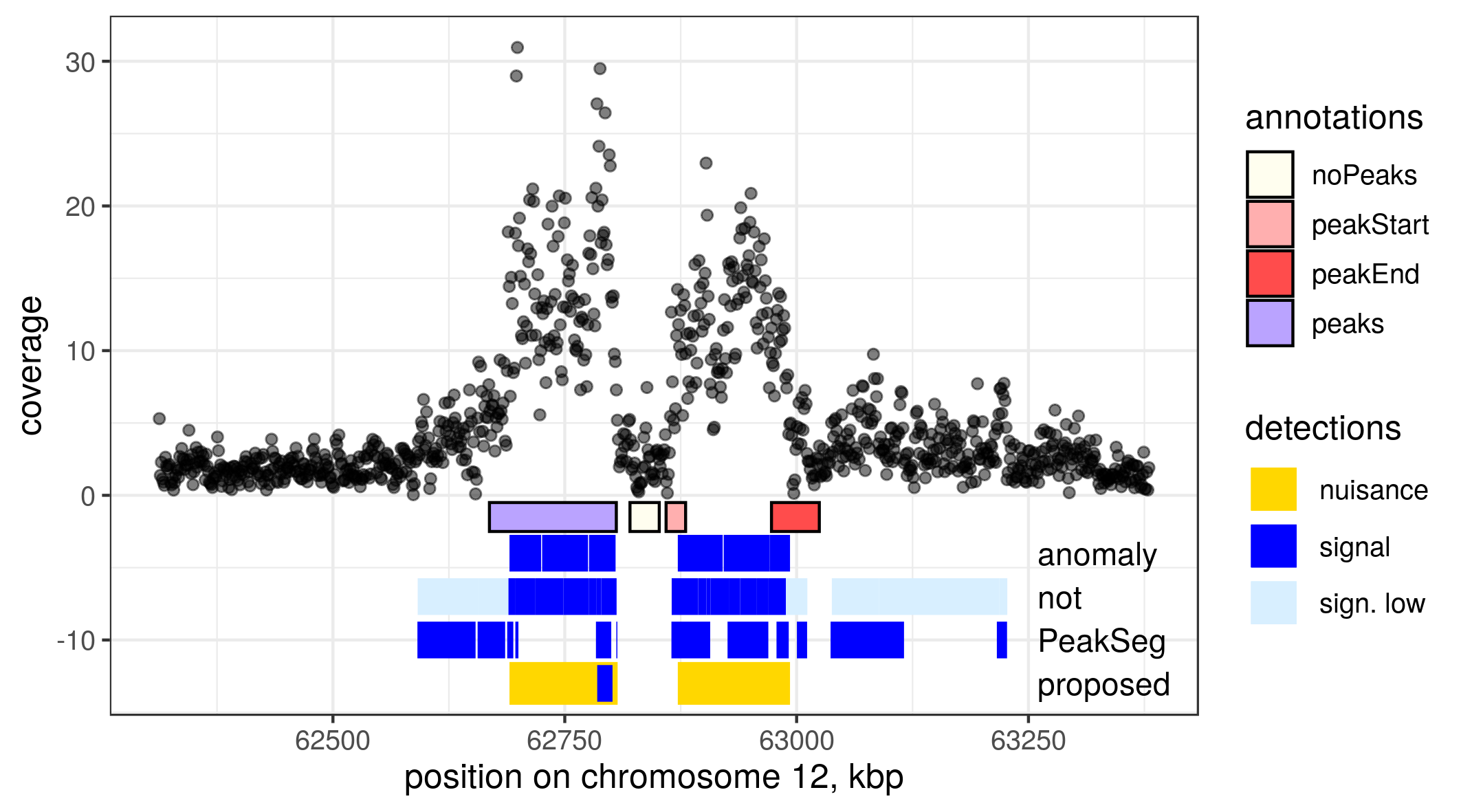}
    \caption{A window of H3K36me3 ChIP-seq data on chromosome 10. Read coverage in 1100 bp windows (black points), manual annotations of peaks included in the dataset (boxes below), and detection results using Algorithm 2 proposed in this paper, as well as three state-of-the-art methods (lines at the bottom).}
    \label{fig:ucichr12}
\end{figure}

\subsubsection{European mortality data}
The recent pandemic of coronavirus disease COVID-19 prompted a renewed interest in early outbreak detection and quantification. In particular, analysis of mortality data provided an important resource for guiding the public health responses to it \citep[e.g.][]{Covid1, Covid2, Covid3}.

We analysed Eurostat data of weekly deaths in Spain over a three year period between 2017 and 2020. Data was retrieved from \url{https://ec.europa.eu} Data Explorer. Besides the impact of the pandemic, mortality data contains normal seasonal variations, particularly in older age groups. We use the 60-64 years age group in which these trends are visually clear and thus provide a ground truth.

Data were analysed with the four methods introduced earlier.
For the \textbf{proposed} and \textbf{anomaly} methods, we used the median and standard deviation of the first 52 weeks of the dataset as estimates of $\mu_0$ and $\sigma_0$, respectively. Penalties in both were set to $3 \log(n)^{1.1}$, as previously. The maximum length of signal segments was set to 10 weeks, to separate seasonal effects. In addition, \textbf{not} was used with default parameters, defining signal as regions where the mean exceeds $\mu_0 \pm \sigma_0$, and \textbf{PeakSegDisk} with a basic grid search to select the penalty as before.

The results of the four analyses are shown in Figure \ref{fig:covid}. Three of the methods, \textbf{anomaly}, \textbf{PeakSeg}, and Algorithm \ref{algfull}, detected a sharp peak around the pandemic period. However, \textbf{anomaly} and \textbf{PeakSeg} also marked one winter period as a signal segment, while ignoring the other two. Four segments were created by \textbf{not}, including a broad peak continuing well past the end of the pandemic spike. In contrast, the proposed method marked the pandemic spike sharply, while also labelling all three winter periods as nuisance segments. The resulting detection using our method is again parsimonious and flexible: if only short peaks are of interest, our method reports those with lower false alarm rate than the other methods, but broader segments are also marked accurately and can be retrieved if relevant.

As in the ChIP-seq data, comparing the results by SIC identifies our method as optimal for this dataset. The values corresponding to \textbf{PeakSeg}, \textbf{not} and \textbf{anomaly} models were 1629.3, 1648.2, and 1626.3 respectively, while Algorithm \ref{algfull} produced a segmentation with SIC 1568.2. Note that the SIC penalizes both signal and nuisance segments, so in this case our model still appears optimal despite having more parameters.

\begin{figure}[h]
    \centering
    \includegraphics[width=0.8\textwidth]{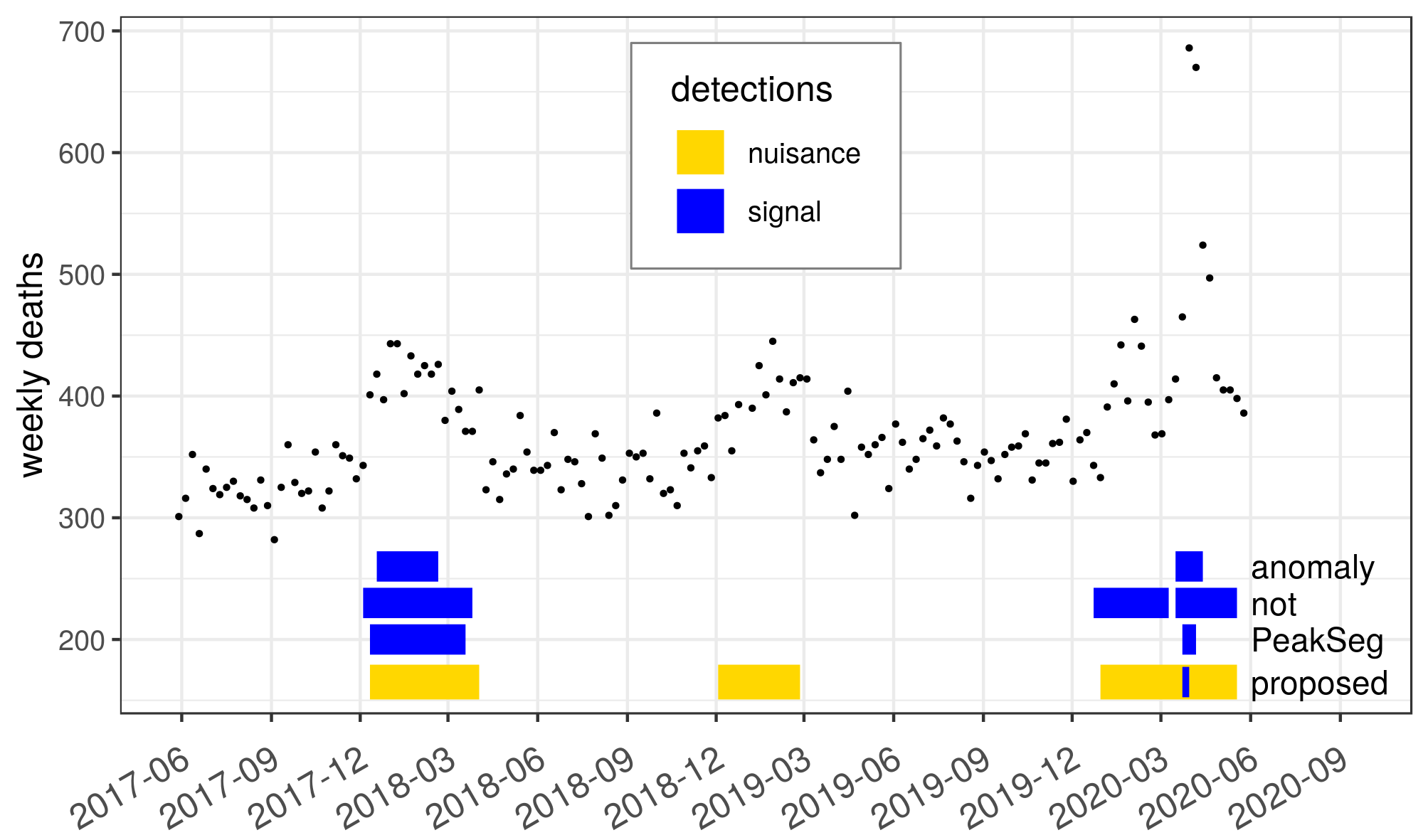}
    \caption{Weekly deaths in Spain, in the 60-64 years age group, over 2017--2020 (black points). Detection results using the method proposed in this paper and three alternative methods shown as lines below.}
    \label{fig:covid}
\end{figure}

\clearpage

\section{Discussion}

In this paper, we have presented a pair of algorithms for improving detection of epidemic changepoints. Similarly to stochastic gradient descent, the iterative updating of the background estimate in Algorithm \ref{algshort} leads to fast convergence while allowing large fraction of non-background points. This is utilised in Algorithm \ref{algfull} to analyse nuisance-signal overlaps.

The computational complexity of both algorithms presented here is $\mathcal{O}(n)$ in the best case, which is similar to state-of-the-art pruned algorithms \citep{KillickPELT, Hocking2018}. However, note that this is stated in the number of required evaluations of $C$. It is usually implicitly assumed that this function can be evaluated and minimised over $\theta$ recursively, so that the total number of operations may also be linear. This would not be achievable with methods that estimate the background level strictly offline, such as aPELT-profile \citep{aPELT2019}. Therefore, development of Algorithm \ref{algshort} was essential to create the overlap detector.

One of major practical benefits of the proposed model is the ability to separate non-target segments. We anticipate that this will greatly improve downstream processing, effectively reducing the false alarm rate or the manual load if the detections are reviewed. Despite that, it is difficult to evaluate this benefit at present: while there are recent datasets with annotations specifically for testing changepoint detection \citep{HockingBio2016, Burg2020}, they are based on labelling all visually apparent changes. In future work, we expect to provide further application-specific comparisons that would measure the impact of separating and neutralising the nuisance process.

\section{Acknowledgements}

This work was supported by the Royal Society of New Zealand Marsden Fund and Te P\=unaha Matatini, a New Zealand Centre of Research Excellence.

\bibliography{biblchangepoint}

\appendix

\section{Proof of Theorem \ref{thm:algconverges} (Convergence of Algorithm \ref{algshort})} \label{app:convergence}

\citet{bottou98} analyses the case of an online algorithm iteratively minimising some function $f(x,w)$ (where $x$ represents the complete data and $w$ the parameters). Data points $\{x_t\}$ arrive sequentially, and at each iteration an estimate $w_t$ of the location of the minimum $w^*$ is obtained using some update function $H(x,w)$ and learning rate $\gamma_t$ as:
\begin{equation}
    \label{eq:sgd}
    w_{t+1} = w_t - \gamma_t H(x_{t+1}, w_t).
\end{equation}
This updating mechanism gives rise to stochastic gradient descent if $\mathbb{E} H(x_{t+1}, w_t) = \nabla_w f(x, w)$, but for the following argument this is not required.

To make the link with Algorithm 1 explicit, the update equation applied by this algorithm can be written as:
\[ w_{t+1} = w_t + \gamma_t (x_{t+1} - w_t). \]

Then $w^* = \theta_0$ (i.e., the background mean that is to be estimated), and we ask whether the sequence of updates converges $w_t \rightarrow w^*$. It was shown by \citet{bottou98} that this occurs almost surely if the following three conditions are met:
\begin{enumerate}
    \item ``convexity'' -- a single optimum $w^*$ exists and the expected value of the updates always points towards it:
        \begin{equation}
            \label{eq:condconv}
            \forall \epsilon>0, \inf_{(w-w^*)^2>\epsilon} (w-w^*) \mathbb{E}H(x, w) > 0;
        \end{equation}
    \item learning rate convergence:
        \begin{equation}
            \label{eq:condlr}
            \sum_{i=1}^\infty \gamma_t = \infty \text{ , } \sum_{i=1}^\infty \gamma^2_t < \infty;
        \end{equation}
    \item bounded variance of the updates:
        \begin{equation}
            \mathbb{E}H(x, w)^2 \leq A + B(w-w^*)^2 \text{ , } A,B \geq 0.
            \label{eq:condeigen}
        \end{equation}
\end{enumerate}

Thus, proof of convergence of our algorithm reduces to showing that these requirements are satisfied. We start with the assumption that the global mean of segment points is also $\theta_0$, and then relax this requirement. 

The following lemma will be needed:
\begin{lemma}
    \label{le:meantr}
    Let $f$ be a unimodal distribution, symmetric around a point $\mu$ (so that $f(x_1) < f(x_2)$ when $x_1 < x_2 \le \mu$ and $f(x_1) > f(x_2)$ when $\mu \le x_1 < x_2$), such as a Gaussian.
    Consider a truncated random variable $X$ with pdf:

    \[ g(x) = \begin{cases} 0 & \text{if } x<m-a \\
            \frac{f(x)}{P(m-a \le x \le m+a)} & \text{if } m-a \le x \le m+a \\
            0 & \text{if } x>m+a
        \end{cases}
    \]
    for some $a>0, m$. Then $\inf_{m \neq \mu}(\mathbb{E}X - m)(\mu - m)>0$.  
\end{lemma}

\begin{proof}
\begin{align}
    \label{eq:meantr}
    \mathbb{E}X - m &= \int_{m-a}^{m+a} (x-m) f(x) dx \nonumber \\ 
    &= \int_{-a}^0 y f(y+m) dy + \int_0^a y f(y+m) dy \nonumber \\ 
    &= \int_0^a y (f(m+y) - f(m-y)) dy.
\end{align}
When $m+a < \mu$, $f$ is increasing throughout the integration range, and $\mathbb{E}X-m>0$; the opposite is true for $m-a > \mu$.
If $m-a < m < \mu < m+a$, split the integral in \eqref{eq:meantr} as:
$$ \mathbb{E}X - m = \int_0^{\mu-m} y (f(m+y) - f(m-y)) dy + \int_{\mu-m}^{a} y(f(m+y) - f(m-y)) dy. $$ 
The first integral covers the range where $f$ is increasing, and thus is positive. Since $\mu-m > 0$, $|m+y-\mu| < |m-y-\mu|$ for $y>0$, and $f(m+y) > f(m-y)$ by symmetry of $f$ around $\mu$ and monotonicity, so the second interval is positive as well. Similarly, $\mathbb{E}X-m < 0$ for $m-a < \mu < m < m+a$. 
\end{proof}

\subsection{When the global mean of segments matches the background mean}

Consider the case that the background points are independent draws from $\mathcal{N}(w^*, \sigma^2)$, so that the points within each segment are $\mathcal{N}(\theta_i, \sigma^2)$, with $\sigma^2$ known, and $\theta \sim \mathcal{N}(w^*, \tau^2)$. Let $w_t$ be the value of the background mean estimated by Algorithm 1 after processing $t$ data points. In this case $w_t \xrightarrow{a.s.} w^*$.

\begin{proof}
Denote the true class of the next data point $x_{t+1}$ by $\delta_{t+1}$ (1 for background points, 0 for signal). Algorithm 1 estimates this as:
\begin{equation*}
    \hat{\delta}_{t+1} = \begin{cases} 1 & \text{if } F(t) + C^0(x_{t+1}; w_t) < \min_{1 \le k \le l} F(t+1-k) + C(x_{t+2-k:t+1}) + \beta \\
        0 & \text{otherwise}.
    \end{cases}
\end{equation*}

Initially, assume for simplicity that the true maximum segment length is 1 (and so only $k=1$ is tested). When $\hat{\delta}_{t+1} = 1$, the background estimate is updated as:
\[ w_{t+1} = w_t + \frac{1}{\sum_{i=1}^t \hat{\delta}_i+1}(x_{t+1} - w_t) \]
(otherwise $w_{t+1} = w_t$). So $\gamma_t = 1/(\sum_{i=1}^t \hat{\delta}_i+1)$, and hence the learning rate convergence conditions \eqref{eq:condlr} are satisfied.

Substituting in the costs based on a one-dimensional Gaussian pdf $\phi$, $\hat{\delta}_{t+1} =1 $ if:
\begin{align}
    &-\log \phi(x_{t+1}; w_t, \sigma^2) < - \log \phi(x_{t+1}; x_{t+1}, \sigma^2) + \beta  \nonumber \\
    \Rightarrow& \frac{1}{\sigma \sqrt{2\pi}} \exp{\frac{ -(x_{t+1}-w_t)^2 }{ 2\sigma^2 }} > \frac{1}{\sigma \sqrt{2\pi}} e^{-\beta} \nonumber \\
    \Rightarrow& | x_{t+1}-w_t | < \sqrt{2 \beta \sigma^2} \nonumber \\
    \Rightarrow& x_{t+1} \in (w_t-\sqrt{2 \beta \sigma^2}; w_t + \sqrt{2 \beta \sigma^2}).
    \label{eq:tru}
\end{align}

The distribution of true segment data points $f(x_{t+1} | \delta_{t+1}=0)$ is symmetric with mean $w^*$ by assumption (as is easily verified for the Gaussian case). The background point distribution is also automatically symmetric.
Thus, the overall distribution of the points used to update the $w_t$ estimate is a truncation of a symmetric unimodal distribution. In the present case, it is a truncated normal with limits $(w_t - \sigma \sqrt{2\beta}, w_t + \sigma \sqrt{2\beta})$, based on \eqref{eq:tru}; more generally it is a truncated variant of the parent distribution with symmetric limits of the form $w_t \pm a$, and parent mean $w^*$ (that the acceptance set is an interval follows from the unimodality of $f$).

This means that $f(x_{t+1} | \hat{\delta}_{t+1}=1)$ satisfies the requirements for Lemma \ref{le:meantr} with $\mu=w^*$, which implies the ``convexity'' condition \eqref{eq:condconv}:
\[ \inf_{(w-w^*)^2>\epsilon} (w-w^*)\mathbb{E}H(x_{t+1}, w) = \inf_{(w-w^*)^2>\epsilon} (w-w^*)(w-\mathbb{E}(x_{t+1} | \hat{\delta}_{t+1}=1)) > 0 \]

Following a similar approach -- conditioning on $\delta_{t+1}$ -- and using the law of total variance it can be shown that the variance of $\mathbb{E}(w-x_{t+1})$ is finite, as required for the condition \eqref{eq:condeigen}, and so $w_t \xrightarrow{a.s.} w^*$.


\textbf{Remark.}
So far, we assumed that segment length $k=1$. If segments occur and are tested in non-overlapping windows of any fixed size $k \ge 2$, the result is similar: $\forall j\in [t-k+1; t], \hat{\delta}_{j+1} =1 $ if:
\[ - \sum_{i=t-k+1}^{i=t} \log \phi(x_{i+1}; w_{t-k}, \sigma^2) < - \sum_{i=t-k+1}^{i=t} \log \phi(x_{i+1}; \bar{x}_k, \sigma^2) + \beta, \]
where $\bar{x}_k = \sum_{i=t-k+1}^{t} x_{i+1} / k$. This can be expressed as truncation limits for accepted $\bar{x}_k$, analogously to \eqref{eq:tru}:
\begin{align}
    \beta &> - \sum_{i=t-k+1}^t \log \phi(x_{i+1}; w_{t-k}, \sigma^2) + \sum_{i=t-k+1}^{i=t} \log \phi(x_{i+1}; \bar{x}_k, \sigma^2) \nonumber \\
    &= \sum_{i=t-k+1}^t \frac{ (x_{i+1}-w_{t-k})^2 }{ 2\sigma^2 } - \frac{ (x_{i+1}-\bar{x}_k)^2 }{ 2\sigma^2 } \nonumber \\
    &= \frac{1}{2\sigma^2} \sum_{i=t-k+1}^t (w_{t-k} - \bar{x}_k)^2 \nonumber \\
    &\Rightarrow \sqrt{2\beta\sigma^2/k} > |w_{t-k} - \bar{x}_k|.
    \label{eq:tru2}
\end{align}

The pdf of $\bar{x}_k$ is a $k$-fold convolution of $f(x)$. Since $f$, as shown earlier, is symmetric unimodal, so is their convolution, and hence the distribution of $\bar{x}_k$, with a mean $\mu=w^*$ \citep{convolutions1998}. In the special case when $f$ is the normal pdf, this can also be shown directly from Gaussian properties.
Then Lemma \ref{le:meantr} implies condition \eqref{eq:condconv}, and the rest of the proof follows as before. 


When all segment positions (overlapping or not) are tested, the background point acceptance rule is:
\[ \hat{\delta}_t = 1 \text{ if } x_t \in \bigcap_{1 \le k \le l} S_{k}, \text{ with } S_{k} = \{x_t: F(t-1) + C^0(x_t) < F(t-k) + C(x_{t-k+1:t}) + \beta \}. \]
As demonstrated earlier, $S_{1} = (w_t - \sqrt{2\beta \sigma^2}, w_t+\sqrt{2\beta \sigma^2}) $.  
Define three sets of $x_t$ based on which rules they pass: $X_1 = x_t \in S_1$, $X_a = x_t \in S_1 \cap S_{\ge 2}$, $X_r = x_t \in S_1 \setminus S_{\ge 2}$, and let $P_a$ and $P_r$ be the probabilities of the corresponding $x_t$ sets. Clearly, $X_1 = X_a \cup X_r$. We are interested in the mean of the points accepted as background, i.e. $\mathbb{E}X_a$.
Assume w.l.o.g. $\mu=0, \sigma=1, w_t>\mu$, as the other case is symmetric. We will now show that for sufficiently large $n$, $\mathbb{E}X_a < w_t$, satisfying \eqref{eq:condconv}.

Using the conditional mean formula:
\begin{align}
        \mathbb{E}X_1 &= P_r \mathbb{E}X_r + P_a \mathbb{E}X_a \nonumber \\
        \mathbb{E}X_a &= \mathbb{E}X_1 / (1-P_r) - P_r \mathbb{E}X_r / (1-P_r).
        \label{eq:Exa}
\end{align}

Assume for now $\mathbb{E}X_r = \mu = 0$.
Then, to obtain $\mathbb{E}X_a < w_t$, we need:
\begin{equation}
    \mathbb{E}X_1 / w_t < 1-P_r.
    \label{eq:tmeanmu0}
\end{equation}

Denote $p = \sqrt{2\beta}$, which is an increasing function of $n$, and consider the growth of both sides of \eqref{eq:tmeanmu0} as $n$ increases.
For the Gaussian or other distributions in the exponential family with mean $\mu$, truncated to a symmetric region $(-a, a)$, it is known that $Var(X | S;\mu) = \frac{d}{d\mu} \mathbb{E}(X|S; \mu)$ \citep{zidek2003}. Then (denoting $S'_1 = (-p; +p)$):
\begin{align*}
    \mathbb{E}X_1 &= w_t + \mathbb{E}(X | S'_1; -w_t) \\
    &= w_t + \mathbb{E}(X | S'_1; 0) + \int_{0}^{w_t} \frac{d \mathbb{E}(X | S'_1; -a)}{ d a } d a \\
    &= w_t + 0 - \int_{0}^{w_t} \frac{d\mathbb{E}(X | S'_1; a) }{ da } d a \\
    &< w_t - w_t \min_{0 \le a \le w_t} Var(X | S'_1;  a)  .
\end{align*}
Hence:
\begin{equation*}
    \mathbb{E}X_1/w_t < 1 - Var(X | S'_1 ; w_t) = 1 - \int_{w_t-p}^{w_t+p} x^2 f(x) dx,
\end{equation*}
where $f(x)$ is the pdf of $x_t$ given $x \in S_1$.
Hence, this side grows with $p$ as $-x^2 f(x)$.

To analyse $P_r$, we first simplify the background acceptance condition.
For any $k \ge 2$, by definition of $F$ and additivity of $C$, we have:
\begin{align*}
    F(t-k) + C(x_{t-k+1:t}) + \beta &= F(t-k) + C(x_{t-k+1:t-1}; \bar{x}_k) + C(x_t; \bar{x}_k) + \beta \\
    & = F(t-k) + C(x_{t-k+1:t-1}; \bar{x}_{-t}) + (k-1) d(\bar{x}_k, \bar{x}_{-t}) + d(\bar{x}_k, x_t) + \beta \\
    & = F(t-1) + A(t-1) + (k-1) d(\bar{x}_k, \bar{x}_{-t}) + d(\bar{x}_k, x_t),
\end{align*}
where $d$ is some distance function, $\bar{x}_{-t}$ is the mean of points $x_{t-k+1:t-1}$, and $A(t) \ge 0$ is a constant depending only on $x_{1:t}$.

It is also helpful to note that $F(t-1) \le F(t-k) + C^0(x_{t-k+1:t-1})$, hence:
\begin{align}
    \label{eq:C00}
    A(t-1) &= F(t-k) + C(x_{t-k+1:t-1}; \bar{x}_{-t}) + \beta - F(t-1) \nonumber \\
    &\ge C(x_{t-k+1:t-1}; \bar{x}_{-t}) + \beta - C^0(x_{t-k+1:t-1}) \nonumber \\
    &\ge \beta - (k-1) d(\bar{x}_{-t}, w_{t-1}).
\end{align}
This corresponds to the case when all $x_{t-k+1:t-1}$ were identified as background.

Using the Gaussian cost, i.e. $d(a,b) = (a-b)^2/2$, and recursive formula for the mean, the acceptance condition for $k$ becomes:
\begin{align}
    \label{eq:acconx}
    F(t-1) + C^0(x_t) &< F(t-1) + A(t-1) + \frac{k-1}{2 k^2} (x_t - \bar{x}_{-t})^2 + \frac{(k-1)^2}{2 k^2} (x_t - \bar{x}_{-t})^2 \nonumber \\
    \Rightarrow (x_t - w_{t-1})^2 &< 2 A(t-1) + \frac{k-1}{k} (x_t - \bar{x}_{-t})^2.
\end{align}

By substituting in the value of $A(t-1)$ from \eqref{eq:C00}, we obtain the following lower bound for $P(x \in S_k | x)$:
\begin{align*}
    P(x \in S_k | x) &\ge P\left( (x_t - w_{t-1})^2 < 2 \beta - (k-1)(\bar{x}_{-t} - w_{t-1})^2 + \frac{k-1}{k} (x_t - \bar{x}_{-t})^2 \right) \\
    &= P\left( (x_t - w_{t-1})^2 - \frac{k-1}{k} (x_t - \bar{x}_{-t})^2  + (k-1)(\bar{x}_{-t} - w_{t-1})^2 < 2 \beta \right)  \\
    &\ge P\left( (x_t - w_{t-1})^2 + (k-1)(\bar{x}_{-t} - w_{t-1})^2 < 2 \beta \right)  \\
    &\ge 1 - \mathbb{E}\left((x_t - w_{t-1})^2) + (k-1)(\bar{x}_{-t} - w_{t-1})^2 \right) / \left( 2 \beta \right) \\
    P(x \notin S_k | x) &\le \mathbb{E}\left((x_t - w_{t-1})^2) + (k-1)(\bar{x}_{-t} - w_{t-1})^2 \right) / \left( 2 \beta \right).
\end{align*}

Thus, we have the following bound for $P_r$ at any $k$:
\begin{align}
    P_r &= \int_{w_t-p}^{w_t+p} P(x \notin S_{k} |x ) f(x) dx \nonumber \\
    &\le \int_{w_t-p}^{w_t+p} O(x^2/p^2) f(x) dx.
    \label{eq:Prejbound}
\end{align}

Therefore, as $N$ increases, $1-P_r$ grows faster than $1-\mathbb{E}X_1/w_t = 1-\int_{w_t-p}^{w_t+p} x^2 f(x) dx$.
This means that $\exists p_0$, and thus $\exists n_0$, such that for $n>n_0$, and thus $p>p_0$, \eqref{eq:tmeanmu0} holds.

So far we assumed $\mathbb{E}X_r = \mu$. Clearly, for larger values of $\mathbb{E}X_r$, $\mathbb{E}X_a$ is even smaller and \eqref{eq:condconv} is satisfied.

For the case when $\mathbb{E}X_r < \mu$, consider the worst case scenario $\mathbb{E}X_r = w_t - p$ (this is the bound to $X_1$, and thus to $X_r$, imposed by $S_1$). Similarly to \eqref{eq:Exa}, we need:
\begin{align*}
   &  \mathbb{E}X_a < w_t  \\
   \Rightarrow& \mathbb{E}X_1/(1-P_r) - P_r (w_t - p)/(1-P_r) < w_t \\
   \Rightarrow& \mathbb{E}X_1/w_t - P_r + P_r p / w_t < 1 - P_r \\
   \Rightarrow& P_r p / w_t < 1 - \mathbb{E}X_1/w_t \\
   \Rightarrow& P_r p / w_t < \int_{w_t-p}^{w_t+p} x^2 f(x) dx.
\end{align*}
However, based on \eqref{eq:Prejbound}, $P_r p / w_t \le \int_{w_t-p}^{w_t+p} O(x^2/p)f(x)dx$, so again $\exists n>n_0$ such that $\mathbb{E}X_a<w_t$, and condition \eqref{eq:condconv} holds.

And so overall $\mathbb{E}(x_t | \hat{\delta}_t=1)$ satisfies condition \eqref{eq:condconv}. Since the distribution of accepted $x_t$ still has bounded support imposed by $S_1$, condition \eqref{eq:condeigen} still holds, and the learning rate condition \eqref{eq:condlr} holds as before, implying convergence.

\end{proof}

\subsection{When the global mean of segments does not match the background mean}
Consider now $\theta \sim \mathcal{N}(\mu, \tau^2)$ for some $\mu \neq w^*$, in particular $\mu >w^*$, so that the overall mean of segment points is $\mathbb{E}(\theta)>w^*$. Then if $w^* < w_t < \mathbb{E}(\theta)$, any segment points that were misclassified as background will (on average) push the estimates away from the background mean, in violation of the ``convexity'' condition \eqref{eq:condconv}.

We assume that each segment point is followed by no less than $n$ background points. Then, as $n \rightarrow \infty$,  $w_t \xrightarrow{a.s} w^*$. For every finite $n$, $\exists \epsilon>0$ such that $P(|w_t - w^*| > \epsilon) = 0 $.

\begin{proof}
Suppose that a misclassification at time $T$ is followed by $n$ correctly classified background points: $\delta_T = 0$, $\delta_t = 1$ for $t \in [T+1; T+n]$, $\hat{\delta}_t = 1$ for $t \in [T; T+n]$. For the points $t \in [T+1; T+n]$, almost sure convergence of $w_t$ was established above, i.e. for all $\epsilon >0$, there exists a $t_0$ such that $\forall t: n \ge t \ge t_0, P(|w_{T+t} - w^*| < \epsilon) = 1$. Therefore, given $n \ge t_0$:
\begin{align}
    \label{eq:primeseqconv}
    & P(|w_{T+n} - w^*| < |w_{T-1} - w^*| ) = 1 \nonumber \\
    \Rightarrow & \begin{cases} P(w_{T+n} - w_{T-1} < 0) = 1, \text{  if } w_{T-1}-w^* > 0 \\
        P(w_{T+n} - w_{T-1} > 0) = 1, \text{  if } w_{T-1} - w^* < 0
    \end{cases} \nonumber \\
    \Rightarrow & \inf_{w_{T-1} \neq w^*} (w_{T-1}-w^*) \mathbb{E}(w_{T+n} - w_{T-1}) < 0.
\end{align}
    
Indexing the segment-background cycles by $i$, denote the first estimate of that segment by $w'_i$, so the set of these estimates are:

    \[ \{w'_i\} = \{w_1, \dots, w_{T-2-n}, w_{T-1}, w_{T+n}, w_{T+1+2n}, \dots\}. \]

The elements of this sequence can be expressed recursively as:
\[ w'_{i+1} = w'_i - \gamma'_i H'(\{x'_i\}, w'_i), \]
with $\{x'_i\} = \{x_t: i(n+1) \le t \le i(n+1)+n \}$.

From \eqref{eq:primeseqconv}, $\mathbb{E}(w'_i - w'_{i+1}) = \gamma'_i \mathbb{E}H'(\{x'_i\}, w'_i)$ is ``convex'' as defined in \eqref{eq:condconv}, and because $\gamma'_i>0$ so is $\mathbb{E}H'(\{x'_i\}, w'_i)$.

Let $\gamma'_i = \frac{1}{i(n+1)+1}$. Then:
\begin{align*}
    H'(\{x'_i\}, w'_i) &= \sum_{t=i(n+1)}^{i(n+1)+n} \gamma_t (w_{t-1} - x_t) / \gamma'_i \\
    &= \sum_{t=i(n+1)}^{i(n+1)+n} \frac{i(n+1)+1}{t+1} (w_{t-1} - x_t) \\
    &< (n+1) (w_{t-1} - x_t).
\end{align*}
So for $n < \infty$, conditions \eqref{eq:condlr}--\eqref{eq:condeigen} are satisfied as well, and $w'_i \xrightarrow{a.s.} w^*$.
(When $n \to \infty$, the convergence conditions are satisfied directly without using the sequence $\{w'_i\}$.)
\end{proof}
    
\subsection{Martingale Approach}
\label{sectmartingales}

We can also describe the update process over the background points using martingales. The algorithm estimates are random variables $w_t$; let $\{\mathcal{W}_t\}$ be the sequence of $\sigma$-algebras such that for each $t$, $w_t$ is measurable with respect to $\mathcal{W}_t$. Using Lemma \ref{le:meantr}, and assuming $w^* < w_t$ again, within each cycle the estimates comprise a supermartingale $\mathbb{E}(w_{t+1} | \mathcal{W}_t) < w_t$ over the points $T \le t < \min(T+n, FHT_w(w^*))$, here $FHT_x(a) = \inf \{t : x_t\le a\}$ is the first hitting time of the process realisation $\{x_t\}$ to value $a$.

Consider again the problematic case when the global mean does not match the background mean and misclassification pushes the estimate away from the background mean, i.e. $w^* < w_{T-1} < w_T < \mathbb{E}(\theta)$.
In order for $w'_i$ to converge, we need the perturbed estimates to return to a value below $w_{T-1}$ in each cycle. At the extremes, we have:
\begin{align*}
    & FHT_w(w_{T}) = T && \text{starting position} \\
    & FHT_w(w_{T-1}) < \infty && \text{for sufficiently large $n$, because } w_t \xrightarrow{a.s.} w^*.
\end{align*}

Clearly, $FHT_w(w_{T-1}) \le FHT_w(w^*)$. However, the number of background points $n$ required to satisfy $FHT_w(w_{T-1}) < T+n$ will depend on five factors: the distribution $f_B$ and penalty $p$ (since they determine the distribution of update values $H$), the size of estimated background set at time $T$ (as it determines the relevant $\gamma_t$), and $w_{T-1}$ and $w_T$.

In practice, $n$ is bounded by the available data, so there is a non-zero probability that, over the segment-background cycles indexed by $i$:
\[ \max_i FHT_w(w_{T_i-1}) > T_i+n. \]

In that case, define $b = \min \{a: \forall i, FHT_w(a) \le T_i+n\}$; the final estimate of $w'_i$ will be bound by $[w^*, b]$. As $n$ increases, $P(|w^*-b|>\epsilon) \rightarrow 0$ for any $\epsilon>0$.

Similar reasoning applies when $\mathbb{E}(\theta) < w_T < w_{T-1} < w^*$.

\section{Proof of Theorem \ref{thm:consist} (Consistency)}
\label{app:consistency}

\begin{proof}
Some general consistency results for changepoint detection by penalised cost methods are given in \citet{CAPA2018}. In particular, an equivalent of our Theorem \ref{thm:consist} is established for any algorithm that detects changepoints by exactly minimising a cost $F(n; \{s_i, e_i\}, \theta, \hat{\mu}, \hat{\sigma})$, where $\hat{\cdot}$ marks robust estimates of background parameters. While the original statement uses the median and interquartile range of $x_{0:n}$ for $\hat{\mu}$ and $\hat{\sigma}$, the proof only requires that the estimates satisfy certain upper bounds on deviations from the true values. Therefore, we will first show that the online estimates produced by Algorithm \ref{algshort} are within these bounds, and then follow the rest of the proof from \citet{CAPA2018}.

Noting again that Algorithm \ref{algshort} is effectively a stochastic gradient descent procedure, with each data point seen precisely once, we can use the error bound on estimates produced by such algorithms as provided in Theorem 7.5 of \citet{SgdTailProb2018}:
\begin{theorem}[\citep{SgdTailProb2018}]
    Let function $f(w)$ be 1-strongly convex and 1-Lipschitz. A stochastic gradient algorithm for minimising this function runs for $T$ cycles, and at each cycle updates the estimate as in (\ref{eq:sgd}) with $\gamma_t = 1/t$, $\mathbb{E}H = \nabla f(w)$. Then:
\[ P\left( \lVert w_t - w^* \rVert ^2 \le O \left(\frac{\log (1/\delta)}{t} \right) \right) \ge 1-\delta. \]
\end{theorem}

Using $\delta = n^{-\epsilon}$, and assuming without loss of generality that $\sigma_0 = 1$, we can establish an upper bound on the error of background parameters estimated by Algorithm \ref{algshort} after $n$ cycles:
\[ P\left( (\hat{\mu} - \mu_0)^2 \le O\left( \frac{\log (n^{\epsilon})}{n} \right) \right) = P\left( |\hat{\mu} - \mu_0| \le O\left(\sqrt{\epsilon}\sqrt{\frac{\log n}{n}} \right) \right) \ge 1-n^{-\epsilon} \]
\[ P\left( |\hat{\sigma}^2 - \sigma^2_0| \le O\left( \sqrt{\epsilon}\sqrt{\frac{\log n}{n}} \right) \right) \ge 1- n^{-\epsilon}. \]
Application of Boole's inequality leads to:
\begin{equation}
    P\left(|\hat{\mu} - \mu_0| \le D_1 \sigma_0 \sqrt{\frac{\log(n)}{n}}, \left| \frac{\hat{\sigma}^2}{\sigma^2_0} - 1 \right| \le D_2  \sqrt{\frac{\log(n)}{n}}\right) \ge 1 - C_1 n^{-\epsilon},
    \label{eq:event8}
\end{equation}
for some constants $C_1, D_1, D_2$ and sufficiently large $n$.
Since the objective function $f$ in our algorithm is the Gaussian log-likelihood (i.e., the updates $\mathbb{E}H$ approximate its gradient), for any given segmentation it is 1-strongly convex. For other functions, overall consistency can still be achieved similarly, but the convergence rate may be slower than $n^{-\epsilon}$.

Having established the bound on estimate errors, we can use Lemma 9 from \citet{CAPA2018} and the proof method reported there.

First, introduce an event $E$ based on a combination of bounds limiting the behaviour of Gaussian data $x_{1:n}$, which for any $\epsilon>0$, occurs with probability $P(E) > 1 - C_2 n^{-\epsilon}$, with some constant $C_2$ and sufficiently large $n$ (Lemmas 1 and 2 in \citet{CAPA2018}). Conditional on this event, the following lemma holds for the epidemic cost $F$ defined as in \eqref{fullcostepid}:

\begin{lemma}[{\citep{CAPA2018}}]
    \label{le:capa9}
    Let $\{\tau\}$ be the set of true segment positions $\{(s_i, e_i)\}$, and $\theta$ the vector of true segment means. Assume $E$ holds, and some $\hat{\mu}, \hat{\sigma}$ are available for which the event in \eqref{eq:event8} holds. Then, there exist constants $C_3$ and $n_1$ such that when $n>n_1$,
    \[ F(n; \{\tau\}, \theta, \hat{\mu}, \hat{\sigma}) - F(n; \{\tau\}, \theta, \mu_0, \sigma_0) < C_4 \log n. \]
\end{lemma}
This lemma, together with results established for classical changepoint detection, can be used to show that the cost of any inconsistent solution will exceed the cost based on true segment positions and parameters (Proposition 8 in \citet{CAPA2018}):
\begin{prop}[\citep{CAPA2018}]
    \label{prop:capa8}
    Define $\{\tau'\}$ to be any set of segments $\{(s_i, e_i)\}$ that does not satisfy the consistency event in (\ref{eq:consist}). Let $\tilde{\theta} = \argmin_{\theta} F(n; \theta) $ be the parameters estimated by minimising the cost for a given segmentation (i.e. the vector of means and/or variances of $x_{s_i:e_i}$ for each $i$). Assume $E$ holds. Then there exist constants $C_4$ and $n_2$ such that, when $n>n_2$:
    \[ F(n; \{\tau'\}, \tilde{\theta}, \hat{\mu}, \hat{\sigma}) \ge F(n; \{\tau\}, \theta, \hat{\mu}, \hat{\sigma}) + C_3 \log(n)^{1+\delta/2}. \]
\end{prop}
See the original publication for a detailed proof of these results.

Finally, for a given set of changepoints, using fitted maximum-likelihood parameters by definition results in minimal cost:
\[ F(n; \{\tau\}, \theta, \hat{\mu}, \hat{\sigma}) \ge F(n; \{\tau\}, \tilde{\theta}, \hat{\mu}, \hat{\sigma}). \]
Thus, when Proposition \ref{prop:capa8} holds, we have:

\[ F(n; \{\tau'\}, \tilde{\theta}, \hat{\mu}, \hat{\sigma}) > F(n; \{\tau\}, \tilde{\theta}, \hat{\mu}, \hat{\sigma}), \]
and an exact minimisation algorithm will always find a solution in the consistent set. 
The overall probability of the events required for Proposition \ref{prop:capa8} is a combination of $P(E)$, established before, and \eqref{eq:event8}, which by Boole's inequality is:
\[ P > 1-C_5 n^{-\epsilon}, \]
for any $\epsilon>0$, $n>n_3$ and some constants $n_3, C_5$.
\end{proof}

\newpage
\section{Proof of Proposition \ref{prop:newpruning} (Pruning of Algorithm \ref{algfull})}
\label{app:pruning}

\begin{proof}
    Denote the true start and end of a nuisance segment as $s_j, e_j$. Consider the case when $s_j \in (t - A(n); t]$. Pruning at time $t$ will not remove this point (i.e. $s_j \notin \bm{k}_{pr,t}$) iff:
    \[ C^0(x_{t-A(n): s_j-1}) + C^N(x_{s_j:t}) < C^0(x_{t-A(n): m-1}) + C^N(x_{m:t}) + \alpha \log(n)^{1+\delta} \]
    with $m$ such that the right hand side is minimised and $m \neq s_j$.

    Denote by $C(x_{a:b}; \hat{\mu}, \hat{\sigma})$ the Gaussian cost calculated with MLE estimates of the parameters (i.e. mean and variance of $x_{a:b}$).
    Note that since $C^0(x_{a:b}) = C(x_{a:b}; \mu_0, \sigma_0)$ and $C^N(x_{a:b}) = C(x_{a:b}; \hat{\mu}, \sigma_N)$, the required event to preserve $s_j$ can be stated as
    \begin{multline*}
        C(x_{t-A(n): s_j-1}; \mu_0, \sigma_0) + C(x_{s_j:t}; \hat{\mu}, \sigma_N) - \\
         C(x_{t-A(n): m-1}; \mu_0, \sigma_0) - C(x_{m:t}; \hat{\mu}, \sigma_N) < \alpha \log(n)^{1+\delta}
    \end{multline*}

    We can establish the probability of this using the following bound (Proposition 4 in \citet{CAPA2018}):
    \begin{lemma}
        \label{le:capa4}
        Let $x_{1:n}$ be piecewise-Gaussian data. Choose any subset $x_{i:j}, 1\le i < j \le n$, with a true changepoint at $s$, i.e., we have $x_t \sim \mathcal{N}(\mu_1, \sigma_1)$ for $t \in [i; s-1]$, and $x_t \sim \mathcal{N}(\mu_2, \sigma_2)$ for $t \in [s; j]$. Then, for any candidate changepoint $\tau$ and any $\epsilon > 0$, there exist constants $B, n_0, K_1$ such that:
        \[ C(x_{i:s-1}; \mu_1, \sigma_1) + C(x_{s:j}; \mu_2, \sigma_2) - C(x_{i:\tau-1}; \hat{\mu}, \hat{\sigma}) - C(x_{\tau:j}; \hat{\mu}, \hat{\sigma}) \le K_1 \log(n) \]
        is true for all $i, j$ with $P \ge 1 - Bn^{-\epsilon}$ when $n>n_0$.
    \end{lemma}

    Now take $i=s_j-A(n)+1, j=s_j+A(n)-1$. Note that there is one and only one changepoint within $x_{i:j}$ because of the required distance between changepoints. Applying Lemma \ref{le:capa4} to such $x_{i:j}$ states that, conditional on an event with probability $P \ge 1 - Bn^{-\epsilon}$, the following is true for all $t \in [s_j, s_j+A(n))$:
    \begin{align*}
         & C(x_{t-A(n): s_j-1}; \mu_0, \sigma_0) + C(x_{s_j:t}; \hat{\mu}, \sigma_N) -
         C(x_{t-A(n): m-1}; \mu_0, \sigma_0) - C(x_{m:t}; \hat{\mu}, \sigma_N)  \\
         & \le C(x_{t-A(n): s_j-1}; \mu_0, \sigma_0)
         + C(x_{s_j:t}; \mu_N, \sigma_N) -
         C(x_{t-A(n): m-1}; \hat{\mu}, \hat{\sigma}) - C(x_{m:t}; \hat{\mu}, \hat{\sigma}) \\
          & \le K_1 \log(n) < \alpha \log(n)^{1+\delta},
    \end{align*}
    where we also used the fact that $\hat{\mu}, \hat{\sigma} = \argmin_{\mu_, \sigma} C(x; \mu, \sigma)$.

    Therefore, with the same probability, $s_j \notin \bigcup_{t=s_j}^{s_j+A(n)-1} \bm{k}_{pr,t}$. Also, $s_j \notin \bigcup_{t=s_j+A(n)}^{n} \bm{k}_{pr,t}$ because then $s_j \le t-A(n)$ and is not considered in the pruning scheme, and clearly $s_j \notin \bigcup_{t=1}^{s_j-1} \bm{k}_{pr,t}$.  
    The case for $e_j$ follows by symmetry, and obviously no true changepoint can be pruned out if $s_j, e_j = \emptyset$, so the overall probability of retaining a true changepoint remains at $P \ge 1 - Bn^{-\epsilon}$.
\end{proof}

\newpage

\section{Supplementary Table}
\setcounter{table}{0}
\renewcommand{\thetable}{S\arabic{table}}

\label{app:suppfigs}

\begin{table}[h]
    \caption{\label{tab:sim2prun} Comparison of Algorithm \ref{algfull} detections without pruning (Full) or with global pruning (Pruned). Of the 5000 simulation runs, all runs where the two options produced any difference in segments were identified. All detections from these runs were extracted and are shown here. Segment types are N -- nuisance, S -- signal.}
    \centering
    \begin{tabular}[htbp]{|lcc|crr|crr|}
        \hline
        & & & \multicolumn{3}{|c|}{(Full)} & \multicolumn{3}{|c|}{(Pruned)} \\
        Scenario & $n$ & Run number & Segm. type & Start  & End & Segm. type & Start & End \\ \hline
        1  &  30  &  88    &    N  &  7 &  21    &    S &   7 &  11 \\
           &  30  &  88    &    S  & 13 &  15    &    N &  12 &  21 \\
           &  30  &  88    &    S  & 16 &  20    &    S &  16 &  20 \\
           & 100  &  84    &    N  & 21 &  69    &    S &  21 &  50 \\
           & 100  &  84    &    S  & 51 &  68    &    S &  51 &  69 \\ \hline
        2  &     30 & 378    &    N &  10 &  18    &    S &  10 &  12 \\ 
           &     30 & 378    &    S &  13 &  15    &    S &  16 &  18 \\
           &     30 & 393    &    N &   2 &  12    &    N &   3 &  12 \\
           &     30 & 393    &    S &   4 &   7    &    S &   4 &   7 \\
           &     30 & 393    &    S &  16 &  18    &    S &  16 &  18 \\
           &     30 & 393    &    S &  22 &  24    &    S &  22 &  24 \\
           &    240 & 445    &    N &  55 &  95    &    N &  55 &  95 \\
           &    240 & 445    &    S & 121 & 144    &    S & 121 & 144 \\
           &    240 & 445    &    N & 170 & 218    &    S & 170 & 192 \\
           &    240 & 445    &    S & 193 & 214    &    S & 215 & 218 \\ \hline
    \end{tabular}
\end{table}


\end{document}